%% file: GV_ver5.tex
\documentclass[onecolumn,11pt]{IEEEtran}
\usepackage{amsmath, amsthm, amssymb}

\usepackage{hyperref}
\newcommand\myshade{70}
\hypersetup{
	linkcolor  = red!\myshade!black,
	citecolor  = blue!\myshade!black,
	urlcolor   = blue!\myshade!black,
	colorlinks = true,
}

\usepackage{latexsym}
\usepackage{graphicx}
\usepackage{verbatim}
\usepackage{mathrsfs}
\usepackage{float}
\usepackage[lofdepth, lotdepth]{subfig}
\usepackage{array}
\usepackage{url}
\usepackage{algorithm}
\usepackage[noend]{algorithmic}
\usepackage[table]{xcolor}
\usepackage{enumerate}
\usepackage{enumitem}
\usepackage{eucal}
\usepackage{stmaryrd}
\usepackage{mathtools}
\usepackage{pgf}
\usepackage{tikz}
\usetikzlibrary{arrows,automata,positioning}
\usepackage[latin1]{inputenc}
\usetikzlibrary{automata}
\usepackage{subfig}

\usepackage{setspace}
\usepackage{adjustbox}
\usepackage{mathtools}

\newcommand{\floor}[1]{\left\lfloor #1 \right\rfloor}

\newcommand{\ras}[1]{\renewcommand{\arraystretch}{#1}}

\theoremstyle{plain}
\newtheorem{theorem}{Theorem}
\newtheorem{proposition}[theorem]{Proposition}
\newtheorem{lemma}[theorem]{Lemma}

\newtheorem{corollary}[theorem]{Corollary}

\theoremstyle{definition}
\newtheorem{definition}{Definition}
\newtheorem{example}[definition]{Example}
\newtheorem{remark}[definition]{Remark}

\newcommand{\cS}{{\mathcal S}}

\DeclareMathAlphabet{\mathbfsl}{OT1}{ppl}{b}{it} 

\newcommand{\bu}{{\mathbfsl u}}
\newcommand{\bv}{{\mathbfsl v}}

\newcommand{\by}{{\mathbfsl y}}

\newcommand{\bk}{{\mathbfsl{k}}}
\newcommand{\bj}{{\mathbfsl{j}}}

\newcommand{\bw}{{\mathbfsl{w}}}

\newcommand{\bx}{{\mathbfsl{x}}}
\newcommand{\bz}{{\mathbfsl{z}}}

\newcommand{\bbZ}{{\mathbb Z}}

\renewcommand{\ge}{\geqslant}
\renewcommand{\le}{\leqslant}

\newcommand{\todo}[1]{{\color{red}(TODO: #1)}}

\newcommand{\bd}{\mathtt{d}}

\newcommand{\HH}{\mathbb{H}}

\newcommand{\bCap}{{\rm Cap}}
\newcommand{\vz}{\mathbfsl{z}}
\newcommand{\vk}{\mathbfsl{k}}
\newcommand{\vm}{\mathbfsl{m}}
\newcommand{\vzero}{\mathbf{0}}
\newcommand{\Hessian}{\mathtt{H}}
\newcommand{\bbS}{\mathbb{S}}
\newcommand{\cT}{\mathcal{T}}

\begin{document}

\title{Gilbert--Varshamov Bound for Codes in $L_1$ Metric using Multivariate Analytic Combinatorics}
\author{
		\IEEEauthorblockN{
			Keshav Goyal\IEEEauthorrefmark{1},
			Duc Tu Dao\IEEEauthorrefmark{1},
			Mladen Kova{\v{c}}evi{\'c}\IEEEauthorrefmark{2}, and 
			Han Mao Kiah\IEEEauthorrefmark{1}}
		
		\IEEEauthorblockA{
			\IEEEauthorrefmark{1}\small School of Physical and Mathematical Sciences, Nanyang Technological University, Singapore\\
			\IEEEauthorrefmark{2}\small Faculty of Technical Sciences, University of Novi Sad, Serbia\\
			Emails: \{keshav002, daoductu001, hmkiah\}@ntu.edu.sg, kmladen@uns.ac.rs\\[-5mm]
		}
  \thanks{The paper was presented in part at the 2023 IEEE International Symposium on Information Theory (ISIT) \cite{Keshav2023}.
  
MK was supported by the Secretariat for Higher Education and Scientific Research of the Autonomous Province of Vojvodina (project number 142-451-2686/2021).}
	}
	
\maketitle

\begin{abstract}
Analytic combinatorics in several variables refers to a suite of tools that
provide sharp asymptotic estimates for certain combinatorial quantities.
In this paper, we apply these tools to determine the Gilbert--Varshamov lower
bound on the rate of optimal codes in $L_1$ metric.
Several different code spaces are analyzed, including the simplex and the
hypercube in $\bbZ^n$, all of which are inspired by concrete data storage and
transmission models such as the sticky insertion channel, the permutation channel,
the adjacent transposition (bit-shift) channel, the multilevel flash memory channel, etc.
\end{abstract}

	

\section{Introduction}
\label{sec:intro}
\input{intro.tex}

\section{Preliminaries}
\label{sec:prelim}
\input{prelim.tex}

\section{GV Bounds for Simplex Spaces} \label{sec:simplex}
\input{simplex.tex}

\section{GV Bound for the Hypercube}
\label{sec:hypercube}
\input{hypercube.tex}

\section{Conclusion}
Obtaining nontrivial lower bounds on the rates of optimal codes in nonuniform spaces is generally a difficult problem.
We have illustrated how the problem can be approached by using the tools of multivariate analytic combinatorics.
In particular, we have derived the general Gilbert--Varshamov bound \cite{gu1993generalized} and its improvements \cite{marcus1992improved} for $L_1$-metric codes in several spaces in $\bbZ^n$.
The study was motivated by concrete communication models, such as the sticky-insertion/tandem-duplication channel, the permutation channel, the adjacent transposition/bit-shift channel, digital communication channels with multilevel baseband transmission, etc.
In addition to being relevant for these and possibly other applications, we hope the paper will inspire further work on the applications of multivariate analytic combinatorics in asymptotic coding theory.

\newpage

\input{GV_copy.bbl}

\newpage
\appendices
\input{appendix.tex}

\end{document}

%% file: intro.tex
Codes in $L_1$ (or Manhattan)\footnote{Recall that the $L_1$ distance between two vectors $\bu=(u_1, \ldots, u_n), \bv=(v_1, \ldots, v_n) \in\mathbb{R}^n$ is defined as $D(\bu,\bv)=\sum_{i=1}^n|u_i - v_i|$.} metric arise as appropriate constructs for error correction in a surprisingly diverse set of applications, in models that at first glance do not have much in common.
We list below only a few examples that served as motivation for the present work.
\begin{itemize}
\item
The sticky insertion (or repetition, or duplication) channel, as well as the $0$-insertion channel that it is equivalent to, were introduced as models for communication in the presence of certain types of synchronization errors.
Codes for these channels \cite{levenshtein1965, dolecek2010repetition, tallini2010, mahdavifar2017} can equivalently be described in the space \cite{kovavcevic2018multi, Sole2018}
\begin{equation}
\label{eq:simplex1}
\Delta^{\!+}_{n,r} = \left\{(u_1,u_2, \ldots, u_r)  \in \mathbb{Z}^r : u_i \ge 1, \sum_{i=1}^r u_i = n\right\}
\end{equation}
under $L_1$ metric, meaning that every code correcting up to $t$ sticky insertions can be obtained from codes in the simplex $\Delta^{\!+}_{n,r}$ having minimum $L_1$ distance $>\!2t$.
In this description, the parameter $n$ corresponds to the length of the input sequences and the parameter $r$ to the number of runs of identical symbols in those sequences.
A very similar characterization holds also in more general channels with uniform tandem duplication errors/mutations that are relevant for in vivo DNA-based data storage systems \cite{jain2017, kovavcevic2018, lenz2018}.
\item
The permutation channel, which randomly reorders the transmitted symbols, has been studied as a model for networks that do not guarantee in-order delivery of packets, as well as for unordered data storage systems, in particular those based on DNA.
Multiset codes \cite{kovavcevic2015, kovavcevic2018multi} that are appropriate for some channels of this kind can equivalently be described in the space
\begin{equation}
\label{eq:simplex2}
\Delta_{n,r} = \left\{(u_1,u_2, \ldots, u_r)  \in \mathbb{Z}^r : u_i \ge 0, \sum_{i=1}^r u_i = n\right\}
\end{equation}
under $L_1$ metric, meaning that a multiset code correcting (e.g.) up to $t$ symbol deletions is equivalent to a code in $\Delta_{n,r}$ having minimum $L_1$ distance $>\!2t$.
Here the parameter $n$ represents the number of transmitted symbols, i.e., code length, and the parameter $r$ the size of the alphabet the symbols take values from.
\item
The binary channel in which input sequences may be impaired by adjacent transpositions (or bit-shifts) was analyzed quite extensively as a model of some magnetic recording devices \cite{shamai1991}.
Codes for such a channel \cite{Vasilev1990, levenshtein1993, Kuznetsov1991, Kolesnik1994} can equivalently be described in the space \cite{kovavcevic2019RLL}
\begin{equation}
\label{eq:nabla}
\nabla_{n,r} = \big\{(u_1,u_2,\ldots, u_r) \in \bbZ^r : 1 \le u_1 < u_2 < \cdots < u_r \le n \big\}
\end{equation}
under $L_1$ metric, meaning that every code correcting up to $t$ adjacent transpositions can be obtained from codes in $\nabla_{n,r}$ having minimum $L_1$ distance $>\!2t$.
In this description, the parameter $n$ corresponds to the length of the input sequences and the parameter $r$ to the Hamming weight of those sequences.
\item
Various types of channels in which the symbols are represented by different voltage/charge levels are of interest in practice, e.g., in digital communication systems employing Pulse Amplitude Modulation, in multilevel flash memories, etc.
The set of all possible inputs in such channels can be represented as
\begin{equation}
\label{eq:cube}
\bbZ_{q}^{n} = \big\{(u_1,u_2,\ldots, u_n) \in \bbZ^n : 0 \le u_i \le q-1 \big\} ,
\end{equation}
where the parameter $n$ corresponds to the code length and the parameter $q$ to the size of the alphabet, i.e., the number of different voltage levels.
It is easy to see that a code correcting a total voltage change of up to $t$ is equivalently described as a code in $\bbZ_q^n$ having a minimum $L_1$ distance $>\!2t$.
See for example \cite{bitouze2010} for an application of the ternary case $\bbZ_{3}^{n}$.
\end{itemize}

In this paper, we study the highest attainable asymptotic rates of codes in the spaces \eqref{eq:simplex1}--\eqref{eq:cube} having minimum $L_1$ distance $d=\delta n$, for any fixed $\delta$ and $n\to\infty$.
In particular, our main contributions include lower bounds on these rates, which, as mentioned in the above examples, can be directly translated into lower bounds on the rates of optimal codes correcting a $\frac{\delta}{2}$-fraction of sticky insertions, transpositions, voltage jumps/drops, etc.
Almost all prior works that discussed bounds on codes for relevant channel models focused on the regime $d=\textnormal{const}$, $n\to\infty$.
To the best of our knowledge, only the work \cite{Kolesnik1994} studied the problem under the assumption $d=\delta n$.
As we shall demonstrate in Section~\ref{sec:nabla}, our bound significantly improves upon the bound from \cite{Kolesnik1994}.

The bounds we derive are versions of the well-known Gilbert--Varshamov (GV) bound \cite{Gilbert1952, Varshamov1957} or, more precisely, of the generalization thereof obtained by Gu and Fuja \cite{gu1993generalized}, which states that the maximum cardinality of a code of minimum distance $d$ is lower-bounded by the ratio of the size of the input space and the \emph{average} volume of a ball of radius $d-1$ in that space.
We shall also follow the approach suggested by Marcus and Roth \cite{marcus1992improved} for further improving this bound.
In order to compute the average ball volume, we employ the tools of multivariate analytic combinatorics (see \cite{Melczer2021} for an introductory text and \cite{pemantle2008} for a survey of combinatorial applications).
We remark that the use of generating functions in determining the GV bound, and in coding theory more generally, is not new.
For example, in one of the pioneering papers, Kolesnik and Krachkovsky \cite{kolesnik1991generating} employed generating functions to compute the GV bound for runlength-limited codes.
However, studies of the multivariate case are very recent \cite{ kovavcevic2019RLL,kovavcevic2021asymptotic,lenz2021multivariate}.
The present paper continues this line of work and we hope it will contribute to inspiring further research on the applications of multivariate analytic combinatorics in coding theory.

\pagebreak
The paper is organized as follows.
In Section~\ref{sec:prelim} we recall the statement of the GV bound and its improvements for general code spaces.
In Section~\ref{sec:ACSV} we give an overview of the results from multivariate analytic combinatorics, and state several refinements thereof, that will be needed in the derivations that follow.
Our main results, GV bounds for codes in the $L_1$ metric, are given in Section~\ref{sec:simplex} for the spaces \eqref{eq:simplex1}--\eqref{eq:nabla} and in Section~\ref{sec:hypercube} for the space \eqref{eq:cube}.
For easier reference, the notation used throughout the paper is summarized in Table~\ref{Notationtable}.

\begin{table}[t]
	\ras{1.5}
	\setlength{\tabcolsep}{4pt}	
	\begin{center}
		\begin{adjustbox}{width=1\textwidth}
			\begin{tabular}{|c|c|c|}
				\hline
				Notation & Remark & Formula\\
				\hline\hline
				$\Sigma$ & Alphabet &  \\
				\hline
				$\cS$ & Constrained (Ambient) Space & \\
				\hline
				$\cS_n$& & $\cS \cap \Sigma^n$  \\
				\hline
				$\bd$ & Metric defined on $\cS$ & \\
				\hline
				$d$ & Distance &  $d= \floor{\delta n}$ \\
				\hline
				$A(\cS_n,d)$ & Maximum code size & \\
				\hline
				$\alpha(\cS,\delta)$ & Highest attainable rate& $\limsup_{n \to \infty} \frac{\log A(\cS_n,\floor{\delta n})}{n}$ \\
				\hline
				$\bCap(\cS)$ & Capacity of $\cS$ &$\limsup_{n \rightarrow \infty} \frac{\log |\cS_n|}{n}$\\
				\hline
				$T(\cS_n,d-1)$ & Total ball  & $\{(\bu,\bv)\in\cS_n \times \cS_n : \bv\in V(\bu,d-1)\}$  \\
				\hline
				$\widetilde{T}(\cS,\delta)$ &  &$\limsup_{n\rightarrow \infty} \frac{\log |T(\cS_n,\floor{\delta n})|}{n}$ \\
				\hline
				$R_\textsc{gv}(\cS, \delta)$ & Standard GV bound  & $2\bCap(\cS) - \widetilde{T}(\cS,\delta)$\\
				\hline
				\hline
				$\cS(\tau)$ & A subset of the Constrained Space & \\
				\hline
				$R_\textsc{mr}(\cS, \delta)$ & GV-MR bound & $\max_{\tau\in I} \big[2\bCap(\cS(\tau))-\widetilde{T}(\cS(\tau), \delta )\big]$\\
				\hline
				\hline
				$\bu, \bv, \bx, \by, \bz, \bw, \bk$ & Vectors & \\
				\hline
				$\bz^\bk$ & Monomial in $\bz, \; \bk \in \bbZ_{\ge 0}^\ell$ & $\prod_{i=1}^\ell z_i^{k_i}$ \\
				\hline 
				$F(\bz), G(\bz), H(\bz)$ & Functions of the vector $\bz$ & \\
				$\bz^*$ & Root of a function & \\
				\hline\hline
				$D(\bu,\bv)$ & $L_1$ distance between $\bu, \bv \in \bbZ^n$ & $\sum_{i=1}^n |u_i - v_i|$ \\
				\hline
				$ \Delta_{n,r}$ & Standard simplex & $\{ (u_1,u_2,\ldots, u_r) \in \bbZ^r: u_i \ge 0, \sum_{i=0}^{r} u_i= n\}$  \\
				\hline
				$ \Delta_{n,r}^{\!+}$ & Positive simplex & $\{ (u_1,u_2,\ldots, u_r) \in \bbZ^r: u_i \ge 1, \sum_{i=0}^{r} u_i= n\}$  \\
				\hline
				$ \nabla_{n,r}$ & Inverted simplex &$\{(u_1,u_2,\ldots, u_r) \in \bbZ^r: 1 \le u_1 < u_2 \ldots < u_r \le n\}$ \\
				\hline
				$ \bbZ_q^n$ & Hypercube & $\{(u_1,u_2,\ldots, u_n) \in \bbZ^n: 0 \le u_i \le q-1 \}$ \\
				\hline\hline
                $\HH(p)$ & Binary entropy function & $-p \log p -(1-p)\log(1-p)$ \\
                \hline
			\end{tabular}	
		\end{adjustbox}
	\end{center}
\caption{}
\label{Notationtable}
\end{table}

%% file: prelim.tex
Let $\Sigma$ be an alphabet, 
$\Sigma^n$ the set of all words of length $n$ over $\Sigma$, and 
$\Sigma^*=\bigcup_{n=1}^\infty\Sigma^n$ the set of all finite-length words over $\Sigma$.
Let $\cS\subseteq \Sigma^*$ be a constrained space and $\cS_n=\cS\cap \Sigma^n$.
Let $\bd: \cS \times \cS \rightarrow \bbZ_{\ge 0}\cup\{\infty\}$ be a metric defined on $\cS$. 
A subset $C \subseteq \cS_n$ such that $\bd(\mathbfsl{c}_1, \mathbfsl{c}_2) \ge d$ for all distinct $\mathbfsl{c}_1, \mathbfsl{c}_2 \in C$ is called an $(\cS_n,d)$-code. The maximum cardinality of a code having a given minimum distance, denoted
\begin{align}
A(\cS_n,d) \triangleq \max\big\{|C|: C \subseteq \cS_n,
\bd(\mathbfsl{c}_1,\mathbfsl{c}_2) \ge\,d \text{ for all } \mathbfsl{c}_1, \mathbfsl{c}_2 \in C, \mathbfsl{c}_1 \neq \mathbfsl{c}_2\big\} ,
\end{align}
is the quantity of central importance in coding theory. In particular, one is interested in determining the highest attainable asymptotic rate,
\begin{align}
\alpha(\cS,\delta)=\limsup_{n \to \infty} \frac{\log A(\cS_n,\floor{\delta n})}{n} ,
\end{align}
for any fixed $\delta$.
An exact characterization of this rate remains elusive in all nontrivial models.
We next describe the best known general lower bound on $\alpha(\cS, \delta)$, which is the main object of study in this paper.

\subsection*{The Gilbert--Varshamov Bound}

For $\bu\in \cS_n$, denote by $V(\bu,r) = \{\bv \in \cS_n: \bd(\bu,\bv) \le r\}$ the ball of radius $r$ centered at $\bu$. 
If $|V(\bu,r)|$ is constant over all $\bu \in \cS_n$, the GV bound states that $A(\cS_n,d) \ge |\cS_n|/|V(\bu,d-1)|$. 
In non-uniform spaces, however, in which $|V(\bu,r)|$ depends on $\bu$, the bound needs to be adapted. 
Kolesnik and Krachkovsky~\cite{kolesnik1991generating} showed that 
the GV lower bound can be generalized to $|\cS|/4\overline{V(d-1)}$, where $\overline{V(d-1)} = \frac{1}{|\cS_n|}\sum_{\bu\in\cS_n} |V(\bu,d-1)| $ is the average ball volume. 
This was further improved by Gu and Fuja~\cite{gu1993generalized} to $|\cS_n|/\overline{V(d-1)}$. 
For convenience, we consider the collection of word pairs
\begin{align}
T(\cS_n,d-1) \triangleq \big\{(\bu,\bv)\in\cS_n \times \cS_n : \bv\in V(\bu,d-1)\big\} .
\end{align}
Hence, $|T(\cS_n,r)|=|\cS_n|\cdot\overline{V(r)}$ represents the ``total ball size'', and the above-mentioned result of Gu and Fuja can be restated as \begin{align}
\label{eq:GVtotalball}
A(\cS_n,d) \ge \frac{|\cS_n|^2}{|T(\cS_n,d-1)|} .
\end{align}

In terms of asymptotic rates when $n\to\infty$, the bound \eqref{eq:GVtotalball} asserts that there exists a family of $(\cS_n,\floor{\delta n})$-codes such that their rates approach 
\begin{equation}
\label{eq:gv}
	R_\textsc{gv}(\cS, \delta) = 2\bCap(\cS) - \widetilde{T}(\cS,\delta)
\end{equation}
where
\begin{equation}
\bCap(\cS) \triangleq \limsup_{n \rightarrow \infty} \frac{\log |\cS_n|}{n}
\end{equation}
and
\begin{equation}
\widetilde{T}(\cS,\delta) \triangleq \limsup_{n\rightarrow \infty} \frac{\log |T(\cS_n,\floor{\delta n})|}{n} .
\end{equation}
Note that $\bCap(\cS) = \widetilde{T}(\cS,0)$.
Therefore, in order to determine $R_\textsc{gv}(\cS, \delta)$, we need to compute  $\widetilde{T}(\cS,\delta)$.

Later, Marcus and Roth \cite{marcus1992improved} improved the GV bound \eqref{eq:gv} by considering certain subsets of the constrained space $\cS$, which are denoted $\cS(\tau)$ for some parameter $\tau$ in a bounded interval $I$; 
we shall refer to this bound as the GV-MR bound. Let $\cS_n(\tau)$ be the set of all words of length $n$ in $\cS(\tau)$ and define $\bCap(\cS(\tau)) = \limsup_{n \rightarrow \infty} \frac{1}{n}\log |\cS_n (\tau)|$. 
Similar to before, define also $\widetilde{T}(\cS(\tau), \delta) = \limsup_{n \rightarrow \infty} \frac{1}{n}\log |T ( \cS_n(\tau),\floor{\delta n})|$. 
Since $\cS_n(\tau)$ is a subset of $\cS_n$, it follows from the usual GV argument that there exists a family of $(\cS_n,\floor{\delta n})$-codes whose rates approach $2\bCap(\cS(\tau))-\widetilde{T}(\cS(\tau), \delta )$ for every $\tau$. 
Therefore, we have the following lower bound on achievable asymptotic code rates:
\begin{equation}
\label{eq:GVMR}
R_\textsc{mr}(\cS, \delta) =  \max_{\tau\in I} \big[2\bCap(\cS(\tau))-\widetilde{T}(\cS(\tau), \delta )\big].
\end{equation}
A key remark from \cite{marcus1992improved} is that both $\bCap(\cS(\tau))$ and $\widetilde{T}(\cS(\tau), \delta )$ can be obtained via different optimization problems.

We refer the reader to \cite{keshav2022evaluating} for a discussion of efficient numerical procedures for solving the optimization problems appearing in the evaluation of the GV and GV-MR bounds.

\section{Analytic Combinatorics in Several Variables (ACSV)}
\label{sec:ACSV}

In many cases of interest, generating functions provide a concise description of the combinatorial quantity $|T(\cS_n,\floor{\delta n})|$ that is needed for determining the GV bound. 
As most of these generating functions involve several variables, 
we will need to borrow tools from multivariate analytic combinatorics to provide the required asymptotic estimates.

Let the number of variables be $\ell$ and let $\vz$ denote the $\ell$-tuple $(z_1,\ldots, z_\ell)$. For $\vk \in \bbZ_{\ge 0}^\ell$, let $\vz^\vk$ denote the monomial $\prod_{i=1}^\ell z_i^{k_i}$. 
Consider a multivariate array of non-negative integers $\{a_{\vk}\}_{\vk \in \bbZ_{\ge 0}^\ell}$ with the generating function $F(\vz) = \sum_{\vk}a_{\vk} \vz^\vk$.
%
%
The following theorem is crucial for this paper.

\begin{theorem}[{\cite[Theorem 1.3]{pemantle2008}}]
\label{thm:acsv}
Let $F(\vz) = \sum_{\vk}a_{\vk} \vz^\vk= \frac{G(\vz)}{H(\vz)}$ where $G$ and $H$ are both analytic, $H(\vzero) \neq 0$, and $a_{\vk} > 0$.
	For each $\vk= (k_1,k_2,\ldots k_\ell)>\vzero$  
	there is a unique solution $\vz^* = (z_1^*, z_2^*, \ldots, z_\ell^*) > \vzero$ 
	satisfying the equations
	\begin{equation}\label{eq:partial}
	\begin{split}
		H(\vz) &= 0\\
		k_\ell z_j \frac {\partial H(\vz)}{\partial z_j}&= k_j z_\ell \frac{\partial H(\vz)}{\partial z_\ell}  \;\text{ for }\; 1 \le j \le \ell -1.
	\end{split}
	\end{equation}
	Furthermore, if $G(\vz^*) \neq 0$,
	\begin{equation}\label{eq:approx}
		a_\vk \sim (2\pi)^{-(\ell-1)/2}\Hessian^{-1/2}\frac{G(\vz^*)}{-z_\ell \frac{\partial H(\vz)}{ \partial z_\ell} \rvert_{\vz =\vz^*}} k_\ell^{-(\ell-1)/2} (\vz^*)^{-\vk},
	\end{equation}
	\noindent {where $\Hessian$ is the determinant of the Hessian of the function parametrizing the hypersurface $\{H=0\}$ in logarithmic coordinates.}
\end{theorem} 
For a detailed calculation of the Hessian matrix, we refer the reader to \cite[Lemma 5]{Melczer2021}. 
More general asymptotic results are available in \cite[Theorems 5.1--5.4]{Melczer2021}. In this paper, we are interested in the case when all coordinates of $\vk$ grow linearly with $n$, i.e., $k_i = nr_i$ for fixed $r_i$, $1\le i\le \ell$. 
In this case all terms in \eqref{eq:approx} tend to constants except $k_\ell^{-(\ell-1)/2} (\vz^*)^{-\vk}$,
and the asymptotic behavior of the sequence $a_\vk$ can be simplified to:
\begin{align}
	a_{\vk} = a_{nr_1,nr_2,\ldots, nr_\ell} &= \Theta\!\left( n^{(\ell-1)/2}\prod_{i=1}^{\ell} (z_i^*)^{-nr_i}\right) \label{eq:asym_approx-0}\\
	\lim_{n \rightarrow \infty} \frac{\log a_{nr_1, nr_2, \ldots, nr_\ell}}{n}  &= -\sum_{i=1}^{\ell} r_i\log z_i^* \label{eq:asym_approx}
\end{align}
To illustrate the theorem, we modify an example from \cite{pemantle2008}. 

\begin{example}[{Binomial coefficients~\cite[Section 4.1]{pemantle2008}}]\label{exa:binomial}
	Consider the bivariate ($\ell=2$) sequence $a_{n,k} = \binom{n}{k}$. The following recursion holds $a_{n,k} = a_{n-1,k-1}+a_{n-1,k}$ for all $(n,k) \neq (0,0)$, from which the generating function $F(\vz) = \frac{1}{1-z_1-z_1z_2}$ can be derived, namely:	
	\begin{align*}
		F(\vz) &= \sum_{n,k \ge 0} a_{n,k}z_1^{n} z_2^{k} \\
		&= 1 + \sum_{(n,k)\ne (0,0)} (a_{n-1,k-1}+a_{n-1,k}) z_1^{n} z_2^{k}\\
		&= 1 + \sum_{(n,k)\ne (0,0)} a_{n-1,k-1} z_1^{n} z_2^{k}+ \sum_{(n,k)\ne (0,0)} a_{n-1,k} z_1^{n} z_2^{k}\\
		&= 1 + (z_1z_2+z_1)F(\vz) .
	\end{align*}
Hence, $G(\vz)=1$ and $H(\vz)= 1-z_1-z_1z_2$.
Further, we can explicitly solve the system of equations \eqref{eq:partial}, which in this example has the form:
	\begin{align*}
		1-z_1-z_1z_2 &= 0, \\ 
		k z_1 \frac{\partial H}{\partial z_1} &= n z_2 \frac{\partial H}{\partial z_1} \,.
	\end{align*}
	From the first equation, we have $z_1+z_1z_2= 1$. Since $z_1 \frac{\partial H}{\partial z_1}= -z_1-z_1z_2 = -1$ and $z_2 \frac{\partial H}{\partial z_1} =-z_1z_2$, we have $k=nz_1z_2$.
	Hence, the solution is $\vz^* = \big(1-\frac{k}{n},\frac{k/n}{1-k/n}\big)$.
	When we fix $k = rn$, we obtain from \eqref{eq:asym_approx} the well-known fact that
	\begin{equation*}
	\lim_{n\rightarrow \infty} \frac{\log a_{n,rn}}{n} = - \log (1-r) - r \log \frac{r}{1-r} 
	=- (1-r) \log (1-r) - r \log r =  \HH(r) .
	\end{equation*}
\end{example}

{

Further, we note that it is possible to reduce the number of variables $\ell$ involved in Theorem~\ref{thm:acsv} using a symmetry argument. We say two dummy variables $z_1$ and $z_2$ are symmetric if $k_1 = k_2$ and interchanging $z_1$ and $z_2$ will have no effect on the generating function $F(\vz)$. Sometimes the symmetric variables can be identified before even finding the generating function just by observing certain symmetry properties of the quantity of interest. Specifically, for the computation of total ball size, it is possible to reduce the number of variables almost by a factor of $1/2$, as will be demonstrated in the next section.

\begin{proposition}\label{prop:barpartial}
	Consider a multivariate polynomial $H(\vz)$ such that in each monomial $\vz^{\vm}$
	we have that $m_1=m_2$. Consider $\vk\in\bbZ_{>0}^\ell$ with $k_1=k_2$. 
	If $\vz^*$ is a solution to \eqref{eq:partial},
	then we have $z_1^* = z_2^*$. 
	We define another multivariate polynomial $\bar{H}$ that involves $\ell-1$ variables. Specifically, we set $\bar{H}(\bar{\vz})=\bar{H}(\bar{z_1},\bar{z_3},\ldots,\bar{z_\ell} )\triangleq H(\bar{z_1},\bar{z_1},\bar{z_3},\ldots,\bar{z_\ell})$.
	 If $\vz^*$ is a solution to \eqref{eq:partial}, 
	 then $\bar{\bz}^*\triangleq(z_1^*,z_3^*,\ldots, z_\ell^*)$ is a solution to the following set of equations	
	\begin{align*}\label{eq:barpartial}
		\bar{H}(\bar{\vz}) &= 0\\
		k_\ell \bar{z_1} \frac{\partial \bar{H}(\bar{\vz})}{\partial \bar{z_1}}&= 2k_1 \bar{z_\ell} \frac{\partial \bar{H}(\bar{\vz})}{\partial \bar{z_\ell}} \\
		k_\ell \bar{z_j} \frac {\partial \bar{H}(\bar{\vz})}{\partial \bar{z_j}}&= k_j \bar{z_\ell} \frac{\partial \bar{H}(\bar{\vz})}{\partial \bar{z_\ell}}  \;\text{ for }\; 3 \le j \le \ell -1.
	\end{align*}
\end{proposition} \begin{proof}
	Since $H(\vz)$ is a multivariate polynomial, we have $H(\vz) = \sum_{\vm}c_{\vm}\prod_{i=1}^{\ell}z_i^{m_i}$. We differentiate $H(\vz)$ with respect to $z_1$ and $z_2$ and multiply by $k_{2}z_{1}$ and $k_{1}z_{2}$, respectively. Further using $m_1 = m_2$ and $k_1 = k_2$, we have
	\begin{align*}
		k_2z_1\frac{\partial H(\vz)}{\partial z_1} = k_2\sum_{\vm}m_1c_{\vm}\prod_{i=1}^{\ell}z_i^{m_i} = k_1\sum_{\vm}m_2c_{\vm}\prod_{i=1}^{\ell}z_i^{m_i} = k_1z_2\frac{\partial H(\vz)}{\partial z_2} .
	\end{align*}
Hence, $z_{1}^* = z_{2}^*$. Furthermore, since  we set $\bar{H}(\bar{\vz})=\bar{H}(\bar{z_1},\bar{z_3},\ldots,\bar{z_\ell} )\triangleq H(\bar{z_1},\bar{z_1},\bar{z_3},\ldots,\bar{z_\ell})$, we have that $\bar{H}(\bar{\vz}) = \sum_{\vm}c_{\vm}\bar{z_{1}}^{2m_1}\prod_{i=3}^{\ell}\bar{z_i}^{m_i}$. Note that, since $\bz^*$ is a solution to \eqref{eq:partial}, we have $H(\bz^*) = 0$ or equivalently $\bar{H}(\bar{\bz^*}) = 0$. Now, we differentiate $\bar{H}(\bar{\vz})$ with respect to $\bar{z_1}$ and $\bar{z_{\ell}}$ and multiply by $k_{\ell}\bar{z_1}$ and $k_1\bar{z_{\ell}}$ respectively:
    \begin{align*}
    k_{\ell} \bar{z_{1}} \frac{\partial \bar{H}(\bar{\vz})}{\partial \bar{z_1}}&= k_{\ell}\sum_{\vm}2m_1 c_{\vm}\bar{z_1}^{2m_1}\prod_{i=3}^{\ell}\bar{z_i}^{m_i} \\ & = 2k_{\ell}\sum_{\vm}m_1 c_{\vm}\bar{z_1}^{2m_1}\prod_{i=3}^{\ell}\bar{z_i}^{m_i} \\ & = 2k_{\ell}\bar{z_1}\frac{\partial H(\bar{z_1},\bar{z_1},\bar{z_3},\ldots,\bar{z_\ell})}{\partial \bar{z_1}} 
     \\ & = 2k_1\bar{z_{\ell}}\frac{\partial \bar{H}(\bar{\vz})}{\partial \bar{z_\ell}}\,.
    \end{align*}
Recall that from \eqref{eq:partial} we have
 \begin{align*}
 	k_{\ell}\bar{z_1}\frac{\partial H(\bar{z_1},\bar{z_1},\bar{z_3},\ldots,\bar{z_\ell})}{\partial \bar{z_1}} & = k_{1}\bar{z_\ell}\frac{\partial H(\bar{z_1},\bar{z_1},\bar{z_3},\ldots,\bar{z_\ell})}{\partial \bar{z_\ell}}
 \end{align*}
and hence
\begin{align*}
  k_{\ell} \bar{z_{1}} \frac{\partial \bar{H}(\bar{\vz})}{\partial \bar{z_1}}&= 2k_1\bar{z_{\ell}}\frac{\partial \bar{H}(\bar{\vz})}{\partial \bar{z_\ell}} .
\end{align*}
The rest of the proof follows directly from \eqref{eq:partial}. 
\end{proof}

 Once we determine the unique solution $\bar{\vz}^* = (z^*,z^*,z_3^*,z_4^*, \ldots, z_{\ell}^*)$ from Proposition~\ref{prop:barpartial}, we have that
 \begin{equation}
   - r_1\log z_1^* - r_2\log z_2^* = -2r\log z_1^* \,.
 \end{equation}
 and therefore, when $k_1=k_2$ or $r_1=r_2$, \eqref{eq:asym_approx} becomes 
 \begin{equation}
 \lim_{n \rightarrow \infty} \frac{\log a_{nr_1, nr_1, nr_3, \ldots, nr_\ell}}{n}  
 =  -\sum_{i=1}^{\ell} r_i\log z_i^* = -2r\log z_1^*-\sum_{i=3}^{\ell} r_i\log z_i^*.  	
 \end{equation}
	
Notice that the unique solution $\vz^*$
in \eqref{eq:partial} is completely represented by $\vk = (k_1, \ldots, k_\ell) = n(r_1, \ldots, r_\ell) = n\boldsymbol{r} $. 
Hence,  the right hand side of \eqref{eq:asym_approx} is a function that depends on $\boldsymbol{r}$.
In what follows, we fix $r_1,\ldots, r_{\ell-1}$ and optimize the expression
with respect to  $r_\ell$. Let $\Phi(\vz^*,r_\ell)$ be the corresponding function.
Let $\vz^*(r_\ell) = (z^*_i(r_\ell))_{i\in[\ell]}$ be the corresponding unique solution in \eqref{eq:partial} while $r_\ell$ varies.
The next theorem characterizes the $\ell$-th component of $\vz^*$ when $\Phi(\vz^*,r_\ell)$ is maximized.
\begin{theorem}\label{thm:maxlog}
	Set $\Phi(\vz^*,r_\ell) \triangleq -\sum_{i=1}^\ell r_i\log z_i^*(r_\ell)$.
	Then $\frac{\partial}{\partial r_\ell} \Phi(\vz^*,r_\ell)=- \log(z_{\ell}^*)$.
	Therefore, the quantity $\Phi(\vz^*,r_\ell)$ is maximized if and only if  
	$z_\ell^*(r_\ell)=1$.
\end{theorem}
\begin{proof}
	As we vary $r_\ell$, since $r_i$ are assumed to be fixed, 
	we have from Theorem~\ref{thm:acsv} that the function $H(\vz^*(r_\ell))$ is identically zero. 
	This implies that its partial derivative with respect to $r_\ell$ is also zero.
	Then applying Chain Rule and \eqref{eq:partial}, we have that  
	\begin{align*}
	\frac{\partial H(\vz^*(r_\ell))}{\partial r_\ell} 
	 =  \sum_{i = 1}^{\ell}\frac{\partial H(\vz^*(r_\ell))}{\partial z_{i}^*}\frac{\partial  z_{i}^*}{\partial r_\ell} 
	 =  \sum_{i = 1}^{\ell}\frac{r_iz_\ell^*}{r_\ell z_i^*}\frac{\partial H(\vz^*(r_\ell))}{\partial z_\ell^*}\frac{\partial z_{i}^*}{\partial r_\ell} 
	 =\frac{z_\ell^*}{r_\ell}\frac{\partial H(\vz^*(r_\ell))}{\partial z_\ell^*}\sum_{i = 1}^{\ell}\frac{r_i}{z_i^*}\frac{\partial z_{i}^*}{\partial r_\ell}
	 =0\,.
	\end{align*}
	Now, since $z_\ell*>0$ and \eqref{eq:partial} implies that the partial derivative with respect to $z^*_\ell$ is nonzero, we have that
	\[	\sum_{i = 1}^{\ell}\frac{r_i}{z_i^*}\frac{\partial z_{i}^*}{\partial r_\ell}
		=0\,.
	\]	
	By taking the partial derivative of $\Phi(\vz^*,r_\ell)$ with respect to $r_\ell$, we have
	\begin{equation*}
		\frac{\partial \Phi(\vz^*,r_\ell)}{\partial r_\ell}  = -\sum_{i=1}^{\ell}\frac{r_i}{z_i^*} \frac{\partial (z_i^*)}{\partial r_\ell} - \log(z_{\ell}^*) =  - \log(z_{\ell}^*).
	\end{equation*}
\end{proof}

\begin{example}[Example~\ref{exa:binomial} continued] 
	Recall that $a_{n,k}=\binom{n}{k}$. Suppose we fix $r_1=1$ and set $k=r_2n$ with $r_2$ varying.
	Then the quantity $\Phi(r_2)=\HH(r_2)$. Theorem~\ref{thm:maxlog} states that $\Phi(r_2)$ is maximized when the solution $z_2^* = r_2/(1-r_2) =1$, i.e., when $r_2=1/2$, as expected.
\end{example}

In the previous section we mentioned that the computation of the GV-MR bound~\cite{marcus1992improved} involves different optimization problems. 
Nevertheless, when we have the explicit generating functions for the constrained space and corresponding ball sizes, we are able to characterize certain component for the solution in~\eqref{eq:partial}.


\begin{corollary}
\label{cor:MRsupport}
Consider a constrained system $\cS$ and its subsets parameterized by $\tau$.
Let $\cS_n(\tau)$ and $T(\cS_n(\tau),\delta n)$ be defined as before, and
suppose that the generating functions
\[
	F_\cS(\vz) = \sum_{\vk\ge 0} |\cS_n(\tau)| \vz^\vk , \quad
	F_T(\by)   = \sum_{\bj\ge 0} |T(\cS_n(\tau),\delta n)| \by^\bj ,
\]
have $\ell$ and $m$ variables, respectively.
Further, suppose that $k_\ell = j_m = \tau n$ and that the other components $k_i$ ($i\in [\ell-1]$) and $j_i$ ($i\in[m-1]$) are all independent of $\tau$.
Then when the optimization program \eqref{eq:GVMR} is maximized, we have that either $(z^*_\ell)^2 = y^*_m$ or $\tau$ is at the end-points of $I$.
\end{corollary}

\begin{proof}
Recall from \eqref{eq:GVMR} that
$R_\textsc{mr}(\delta) = \max_{\tau\in I} [2\bCap(\cS(\tau))-\widetilde{T}(\cS(\tau), \delta)]$.
Furthermore, $\bCap(\cS(\tau)) = \Phi(\vz^*,\tau)$ and  $\widetilde{T}(\cS(\tau), \delta) = \Phi(\by^*,\tau)$, as defined in Theorem~\ref{thm:maxlog}. It follows from Theorem~\ref{thm:maxlog} that 
\begin{align*}
	\frac{\partial}{\partial \tau} \Phi(\vz^*,\tau) = - \log(z_\ell^*) , \quad
	\frac{\partial}{\partial \tau} \Phi(\by^*,\tau)   = - \log(y_m^*).
\end{align*}
Hence, to maximize the right-hand side of \eqref{eq:GVMR}, we need $2 \frac{\partial}{\partial \tau} \Phi(\vz^*,\tau) - \frac{\partial}{\partial \tau} \Phi(\by^*,\tau)= 0$ or, equivalently, $(z^*_\ell)^2 = y^*_m$.
\end{proof}

}

%% file: simplex.tex
\subsection{The Standard Simplex}
\label{subsec:US}

In this subsection, we consider the code space given by the {\em simplex} of dimension $r$ and weight $n$, namely
\begin{equation}
	\Delta_{n,r} \triangleq \left\{\bu= (u_1,u_2,\ldots, u_r) \in \bbZ^r : u_i \ge 0, \; \sum_{i=0}^{r} u_i= n\right\} .
\end{equation}
$\Delta_{n,r}$ can be understood as the set of all multisets of cardinality $n$ over an alphabet of size $r$. As noted in the Introduction, the motivation for studying codes in this space comes from certain types of permutation channels and unordered (e.g., DNA-based) data storage channels.

We shall consider the asymptotic regime where $n\to\infty$ and $r=\rho n$, for an arbitrary constant $\rho\ge0$.
Denote the family of simplices satisfying this relation by $\Delta_\rho$, i.e.,
\begin{equation}
\Delta_\rho = \bigcup_{n=1}^\infty \Delta_{n,\floor{\rho n}} .
\end{equation}
The following statement is well-known and follows directly from the fact that $|\Delta_{n,r}| = \binom{n+r-1}{r-1}$.

\begin{proposition}\label{prop:run-cap-L1}
For fixed $\rho \ge 0$,
\[
\bCap(\Delta_\rho) \triangleq \lim_{n\to\infty} \frac{\log |\Delta_{n,\floor{\rho n}}|}{n} = (1+\rho)\HH\Big(\frac{\rho}{1+\rho}\Big).\]
\end{proposition}


\subsubsection{Evaluation of the GV Bound}\label{sec:TotalBallSize}

We denote the $L_1$-ball with center $\bu \in \Delta_{n,r}$ and radius $s$ by $V(\bu,s) \triangleq \{\bv \in \Delta_{n,r}: D(\bu,\bv) \le s\}$.
To determine the total ball size, we first consider the number of pairs $(\bu,\bv)$ with $L_1$-distance $s$, denoted by
\begin{equation}
N(n_1,n_2,r,s) = \left| \left\{(\bu,\bv) \in \Delta_{n_1,r}\times \Delta_{n_2,r}: D(\bu,\bv) =  s \right\} \right| .
\end{equation}
This quantity can be recursively characterized as follows.

\begin{lemma}
\label{lem:general_l1}
\begin{align*}
  N(n_1,n_2,r,s) &= \sum_{i \ge 0}N(n_1-i,n_2-i,r-1,s) \\ 
		&\hspace{4mm} + \sum_{i \ge 0}\sum_{j \ge 1}N(n_1-i,n_2-i-j,r-1,s-j) \\
		&\hspace{4mm} + \sum_{i \ge 0}\sum_{j \ge 1}N(n_1-i-j,n_2-i,r-1,s-j)\,.  
\end{align*}
\end{lemma}
\begin{proof}
Let $\bu= (u_1,\ldots, u_r) \in \Delta_{n_1,r}$ and $\bv=(v_1,\ldots,v_r) \in \Delta_{n_2,r}$. Consider truncating the last coordinate $u_r$ and $v_r$. If $u_r=v_r=i$, we get the first sum where the distance between the truncated vectors remains the same. Otherwise, suppose $u_r = i$ and $v_r = i+j$, for $i \ge 0$ and $j \ge 1$. In this case, the truncated vectors belong to $\Delta_{n_1-i,r-1}$ and $\Delta_{n_2-i-j,r-1}$, respectively, and the distance between them is $s - |u_r-v_r| = s - j$, which gives the second term. The last term is obtained similarly when $u_r > v_r$.
\end{proof}

We are ready to derive the generating function $F(x_1,x_2,y,z)=\sum_{n_1,n_2,r,s \ge 0}N(n_1,n_2,r,s)x_{1}^{n_{1}}x_{2}^{n_{2}}y^rz^s$. In order to apply Theorem~\ref{thm:acsv}, we present $F(x_1,x_2,y,z)$ as a ratio of two multivariate polynomials.

\begin{lemma}
\label{lem:general_ge}
$F(x_1,x_2,y,z) = \frac{G(x_1,x_2,y,z)}{H(x_1,x_2,y,z)}$, where $G(x_1,x_2,y,z)$ is some multivariate polynomial and 
\begin{align*}
H(x_1,x_2,y,z) &=(1-x_1)(1-x_2)\big[(1-x_1x_2)(1-x_1z)(1-x_2z) - y(1-x_1x_2z^2)\big]\,.  
\end{align*}
\end{lemma}

\begin{proof}
Applying Lemma~\ref{lem:general_l1}, we obtain
\begin{align*}
			F(x_1,x_2,y,z)
			=\ &\sum_{n_1,n_2,r,s \ge 0}N(n_1,n_2,r,s)x_{1}^{n_{1}}x_{2}^{n_{2}}y^rz^s \\ 
			=\ &\sum_{n_1,n_2,r,s \ge 0}\sum_{i \ge 0}N(n_1-i,n_2-i,r-1,s)x_{1}^{n_1}x_{2}^{n_2}y^rz^s \\ 
			& + \sum_{n_1,n_2,r,s \ge 0}\sum_{i \ge 0}\sum_{j \ge 1}N(n_1-i,n_2-i-j,r-1,s-j)x_{1}^{n_1}x_{2}^{n_2}y^rz^s \\ 
			& + \sum_{n_1,n_2,r,s \ge 0}\sum_{i \ge 0}\sum_{j \ge 1}N(n_1-i-j,n_2-i,r-1,s-j)x_{1}^{n_1}x_{2}^{n_2}y^rz^s
			\\ =\ & 1 + \sum_{i \ge 1} x_{1}^{i} + \sum_{i \ge 1} x_{2}^{i} + F(x_1,x_2,y,z)\bigg(y\sum_{i \ge 0}(x_1x_2)^i\bigg)\bigg(1 + \sum_{j \ge 1}(x_1z)^j + \sum_{j \ge 1}(x_2z)^j\bigg)\,.
\end{align*}
Hence,
\begin{align*}
			\hspace{-1mm} F(x_1,x_2,y,z)
			&= \frac{1 + \sum_{i \ge 1} x_{1}^{i} + \sum_{i \ge 1} x_{2}^{i}}{1-\big(y\sum_{i \ge 0}(x_1x_2)^i\big)\big(1 + \sum_{j \ge 1}(x_1z)^j + \sum_{j \ge 1}(x_2z)^j\big)} \\ &= \frac{(1 + \sum_{i \ge 1} x_{1}^{i} + \sum_{i \ge 1} x_{2}^{i})(1-x_1x_2)(1-x_1z)(1-x_2z)}{(1-x_1x_2)(1-x_1z)(1-x_2z) - y(1-x_1x_2z^2)}\,.\\
			& = \frac{(1-x_1x_2)^2(1-x_1z)(1-x_2z)}{(1-x_1)(1-x_2)[(1-x_1x_2)(1-x_1z)(1-x_2z) - y(1-x_1x_2z^2)]} \,.
\end{align*}
\end{proof}

As we are interested in the pairs residing in the simplex with the same weight $n$ and dimension $\rho n$, and having distance $\delta n$, the quantity of interest is  $N(n,n,\rho n, \delta n)$. Since $n_1 = n_2 = n$, we can reduce the number of variables in the generating function using Proposition~\ref{prop:barpartial} and set $x_1 = x_2 = x$.
More specifically, we define the generating function
\begin{equation}
  F(x,x,y,z) = \sum_{n,r,s \ge 0} N(n,n,r,s) x^{2n}y^rz^s
             = \frac{(1+x)^2(1-xz)}{(1-x^2)(1-xz) - y(1+xz)} .
\end{equation}
By Theorem~\ref{thm:acsv} and Proposition~\ref{prop:barpartial}, with $ r= \rho n, s = \delta n$, we need to solve the following system of equations: 
\begin{equation}\label{eq:L1-H}
	H = 0 \text{\quad and\quad}
	\frac{xH_x}{2} = \frac{yH_{y}}{\rho} = \frac{zH_{z}}{\delta}, \;\;\text{where}\;\; H= (1-x^2)(1-xz) - y(1+xz) .
\end{equation}
Here $H_x$ denotes the partial derivative $\frac{\partial H}{\partial x}$, and the function $H(x,y,z)$ is written as $H$ for simplicity.

\begin{lemma}
\label{lem:solveH_L1_1}
	For fixed $\rho$, the solution of \eqref{eq:L1-H} with respect to $\delta$ is
		\begin{align*}
			x^*(\delta) &= \sqrt{1 - \frac{2\rho}{2+2\rho-\delta}},\\
			z^*(\delta) &= \frac{\sqrt{\rho^2 + \delta^2} - \rho}{x^*\delta},\\
			y^*(\delta) &= 2\frac{\sqrt{\rho^2 + \delta^2}-\delta}{2-\delta + 2\rho} \,. 
		\end{align*}
\end{lemma}
\begin{proof}
Appendix~\ref{appen_solve_L1}.
\end{proof}

Applying \eqref{eq:asym_approx}, we have that
\begin{equation}
\label{eq:asym_ball}
	\lim_{n \to \infty} \frac{\log N(n,n,\rho n,\delta n)}{n} = -2\log x^*(\delta) - \rho \log y^*(\delta) - \delta \log z^*(\delta)\,.  
\end{equation}
The total ball size is
\begin{equation*}
|T(\Delta_{n, \rho n}, \delta n)| = \sum_{s=0}^{\delta n} N(n,n,\rho n, s)
\end{equation*}
and hence 
\begin{align*}
\widetilde{T}(\Delta_\rho,\delta) 
&\triangleq \lim_{n\rightarrow \infty} \frac{\log |T(\Delta_{n, \rho n}, \delta n)|}{n} \\
	& = \max_{0\le \delta_1 \le \delta} -2\log x^*(\delta_1) - \rho \log y^*(\delta_1) - \delta_1 \log z^*(\delta_1) \\
	& = \begin{cases}
		-2\log x^*(\delta) - \rho \log y^*(\delta) - \delta \log z^*(\delta)\,, & \delta\le \delta_{\rm{max}}\,,\\
		2(1+\rho)\HH(\frac{\rho}{1+\rho})\,, & \delta\ge \delta_{\rm{max}}\,.
	\end{cases}
\end{align*}
Here, $\delta_{\rm{max}}= 2(1+\rho)/(2+\rho)$.
Details of the derivations for $\delta_{\rm{max}}$ are also given in Appendix~\ref{appen_solve_L1}.

In conclusion, we have the following explicit formulas for the asymptotic total ball size and the GV lower bound.

\begin{proposition}
\label{cor:rate_L1}
For fixed $\rho\ge0$, set $\delta_{\rm{max}}= 2(1+\rho)/(2+\rho)$.
When $0\le\delta\le\delta_{\rm{max}}$, we have
\begin{align*}
	\widetilde{T}(\Delta_\rho,\delta) = &-\rho + \delta \log(\delta) - \rho \log\!\big(\sqrt{\rho^2+\delta^2} - \delta\big)
	- \delta \log\!\big(\sqrt{\rho^2+\delta^2}-\rho\big) \\
	& +(1+\rho-\delta/2) \log(2+2\rho-\delta) - (1-\delta/2)\log(2-\delta) \,.
\end{align*}
%
Otherwise, when  $\delta\ge\delta_{\rm{max}}$, we have $\widetilde{T}(\Delta_\rho,\delta)=2(1+\rho)\HH(\frac{\rho}{1+\rho})$.
\end{proposition}

\begin{theorem}
\label{prop:gv-L1}
	For fixed $\rho, \delta \ge 0$, we have $\alpha(\Delta_{\rho},\delta)\ge R_{\textsc{gv}}(\Delta_\rho, \delta)$, where
	$$R_{\textsc{gv}}(\Delta_\rho, \delta) \triangleq 2(1+\rho)\HH\Big(\frac{\rho}{1+\rho}\Big) - \widetilde{T}(\Delta_\rho,\delta) .$$
\end{theorem}

\subsubsection{Evaluation of Marcus and Roth's Improvement of the GV Bound}
In this subsection, we further improve the bound obtained above by introducing an additional parameter that constrains the code space, then deriving the GV bound for the resulting space, and finally optimizing this bound over all values of the introduced parameter.
This approach was first suggested in \cite{marcus1992improved}.

The additional parameter we introduce, denoted $p$, is the number of zeros in each vector. More precisely, we set
\begin{equation}
\Delta_{n,r}(p) = \left\{\bu=(u_1,u_2,\ldots, u_r) \in \bbZ^r: u_i \ge 0, \sum_{i=0}^r u_i=n, |\{i : u_i=0\}| =p \right\} .
\end{equation}
We allow $p$ to grow linearly with $n$ and set  $p=\tau n$.
Let $\Delta_\rho(\tau)$ denote the family of constrained spaces satisfying this relation, i.e., $\Delta_\rho(\tau)=\bigcup_n \Delta_{n,\floor{\rho n}}(\lfloor \tau n \rfloor)$.

\begin{proposition}
\label{prop:run-capMR-L1}
For fixed $\rho \ge 0$, $0 \le \tau \le \rho$, we have $$\bCap(\Delta_\rho(\tau)) \triangleq \lim_{n\to\infty}\frac{\log |\Delta_{n,\floor{\rho n}}(\lfloor \tau n \rfloor)|}{n}=\rho \HH(\tau/\rho) + \HH(\rho-\tau).$$
\end{proposition}
\begin{proof}
Follows directly from $|\Delta_{n,r}(p)| = \binom{r}{p}\binom{n-1}{r-p-1}$.
\end{proof}

We perform the same analysis for total ball size to obtain the bound.
We consider balls centered at $\bu \in \Delta_{n,r}(p)$ and having radius $s$, that is 
$V(\bu,s,p) = \{\bv \in \Delta_{n,r}(p): D(\bu,\bv) \le s\}$.

To estimate the ball size, we consider the number of pairs $(\bu,\bv)$ with $L_1$-distance exactly $s$, denoted by $N(n_1,n_2,r,s,p_1,p_2) = |\{(\bu,\bv) \in \Delta_{n_1,r}(p_1)\times \Delta_{n_2,r}(p_2): D(\bu,\bv) = s, |\{i : u_i = 0\}| = p_1, |\{i : v_i = 0\}| = p_2\}|$. The following lemma gives a recursive expression for this quantity.

\begin{lemma}
\label{lem:general_l1_1}
\begin{align*}
N(n_1,n_2,r,s,p_1,p_2) =\ & N(n_1,n_2,r-1,s,p_1-1,p_2-1) \\ 
		&+ \sum_{j \ge 1}N(n_1,n_2-j,r-1,s-j,p_1-1,p_2) \\
		&+ \sum_{j \ge 1}N(n_1-j,n_2,r-1,s-j,p_1,p_2-1) \\				
		&+ \sum_{i \ge 1}N(n_1-i,n_2-i,r-1,s,p_1,p_2) \\ 
		&+ 2\sum_{i \ge 1}\sum_{j \ge 1}N(n_1-i,n_2-i-j,r-1,s-j,p_1,p_2) \\
		&+ 2\sum_{i \ge 1}\sum_{j \ge 1}N(n_1-i-j,n_2-i,r-1,s-j,p_1,p_2) \,. 
\end{align*}
\end{lemma}
\begin{proof}
Let $\bu= (u_1,\ldots, u_r) \in \Delta_{n_1,r}(p_1)$ and $\bv=(v_1,\ldots,v_r) \in \Delta_{n_2,r}(p_2)$. Truncate the last coordinate $u_r$ and $v_r$, and consider the following possibilities:
(1) $u_r=v_r=0$,
(2) $u_r = 0$, $v_r = j$,
(3) $u_r = j$, $v_r = 0$,
(4) $u_r = i$, $v_r = i$,
(5) $u_r = i$, $v_r = i+j$,
(6) $u_r = i+j$, $v_r = i$.
\end{proof}

As before, the quantity of interest is  $N(n,n,\rho n, \delta n, \tau n, \tau n)$. Therefore, we define the generating function $F(x,y,z,w)=\sum_{n,r,s,p \ge 0}N(n,n,r,s,p,p)x^{2n}y^rz^sw^{2p}$ with four instead of six variables (see Proposition~\ref{prop:barpartial}).
We next show that $F$ is rational and that Theorem~\ref{thm:acsv} can be applied to it.

\begin{lemma}
\label{lem:general_ge_mr}
	$F(x,y,z,w) = \frac{G(x,y,z,w)}{H(x,y,z,w)}$, where $G(x,y,z,w)$ is some multivariate polynomial and
	{
	    \begin{align*}
	    	H(x,y,z,w) &= (1-x^2)(1-xz) - y(w^2(1-xz)(1-x^2) + 2wxz(1-x^2) + x^2(1+xz))\,.  
	    \end{align*}
	}
\end{lemma}
\begin{proof}
	{
		\begin{align*}
			F(x,y,z,w)
			&= \sum_{n,r,s,p \ge 0}N(n,n,r,s,p,p)x^{2n}y^rz^sw^{2p} \\ 
			&= \sum_{n,r,s,p \ge 0}N(n,n,r-1,s,p-1,p-1)x^{2n}y^rz^sw^{2p} \\ 
			& \hspace{4mm} + \sum_{n,r,s,p \ge 0}\sum_{j \ge 1}N(n,n-j,r-1,s-j,p-1,p)x^{2n}y^rz^sw^{2p} \\
			& \hspace{4mm} + \sum_{n,r,s,p \ge 0}\sum_{j \ge 1}N(n-j,n,r-1,s-j,p,p-1)x^{2n}y^rz^sw^{2p} \\  
			& \hspace{4mm} + \sum_{n,r,s,p \ge 0}\sum_{i \ge 1}N(n-i,n-i,r-1,s,p,p)x^{2n}y^rz^sw^{2p} \\ 
			& \hspace{4mm} + \sum_{n,r,s,p \ge 0}\sum_{i \ge 1}\sum_{j \ge 1}N(n-i,n-i-j,r-1,s-j,p,p)x^{2n}y^rz^sw^{2p} \\
			& \hspace{4mm} + \sum_{n,r,s,p \ge 0}\sum_{i \ge 1}\sum_{j \ge 1}N(n-i-j,n-i,r-1,s-j,p,p)x^{2n}y^rz^sw^{2p} \\  
		 &=  1 + 2\sum_{i \ge 1} x^{i} + F(x,y,z,w)y\Bigg(w^2 + 2w\sum_{j \ge 1}(xz)^j + \sum_{i \ge 1}(x^{2i}) + 2\sum_{i \ge 1}\sum_{j \ge 1}(x^{2i})(xz)^j\Bigg)\,.
		\end{align*}
	}
	Hence,
	{
		\begin{align*}
			\hspace{-1mm} F(x,y,z,w)
			&= \frac{1 + 2\sum_{i \ge 1} x^{i}}{1-y\big(w^2 + 2w\sum_{j \ge 1}(xz)^j + \sum_{i \ge 1}(x^{2i}) + 2\sum_{i \ge 1}\sum_{j \ge 1}(x^{2i})(xz)^j\big)} \\ 
			&= \frac{(1-xz)(1+x)^2}{(1-x^2)(1-xz) - y\big(w^2(1-xz)(1-x^2) + 2wxz(1-x^2) + x^2(1+xz)\big)} \,.
		\end{align*}
	}
\end{proof}

From Theorem~\ref{thm:acsv} and Proposition~\ref{prop:barpartial}, we need to solve the following system of equations: 
\begin{equation}\label{eq:L1MR-H}
	H = 0 \text{\quad and\quad}
	\frac{xH_{x}}{2} = \frac{yH_{y}}{\rho} = \frac{zH_{z}}{\delta}  = \frac{wH_{w}}{2\tau} ,
\end{equation}
where the function $H(x,y,z,w)$ is written as $H$ for simplicity. 

Let $\lambda_1 = xz$ and $\lambda_2 = x^2$.

\begin{lemma}
\label{lem:solveH_L1}
	When $\rho$ is fixed, the solution of (\ref{eq:L1MR-H}) with respect to $\delta$ is:
	{\begin{align*}
		x^*(\delta)  &=  \sqrt{\lambda_2^*},\\
		z^*(\delta) &= \frac{\lambda_1^*}{x^*},\\
		w^*(\delta) &= \frac{\lambda_2^*(\delta(1-(\lambda_1^*)^2)-2\lambda_1^*(\rho-\tau))}{\lambda_1^*(1-\lambda_2^{*})(2(\rho - \tau) - \delta(1-\lambda_1^*))}                       , \\
		y^*(\delta) &= \frac{(1-\lambda_1^*)(1-\lambda_2^*)}{(w^*)^2(1-\lambda_1^*)(1-\lambda_2^*) + 2w^*\lambda_1^*(1-\lambda_2^*) + \lambda_2^*(1+\lambda_1^*)}  
		\end{align*}
		where $\lambda_1^*$ is the root of the equation
		\small{
		\begin{align*}
			& \delta^2(1+\rho)\lambda_1^5 + (-\delta^2(1+\tau)+4\delta(\rho-\tau)(1+\rho))\lambda_1^4 + (-\delta^2(\rho+2)+2\delta(\rho-\tau)(\rho-\tau-2)+4(\rho-\tau)^2(1+\rho))\lambda_1^3 \\ & + (-\delta^2(\rho-2\tau-2)-2\delta(\rho-\tau)(\rho+\tau+2)-4(\rho-\tau)^2(1-\rho))\lambda_1^2 
			+ (\delta^2 - 4\delta(\rho-\tau)(\rho-\tau-1))\lambda_1 + \delta^2(\rho-\tau-1) = 0 \,\,.
		\end{align*}
	}
		and $\lambda_2^*$ is given by $$\lambda_2^*(\delta) = 1 - (1+\lambda_1^*)\frac{2(\rho - \tau) - \delta(1-\lambda_1^*)}{(2-\delta)}.$$

	}
\end{lemma}
\begin{proof}
Appendix~\ref{appen_solve_L1_MR}. 	
\end{proof}

Applying \eqref{eq:asym_approx}, we have that
\begin{equation}\label{eq:L1-asym_ball}
	\lim_{n \to \infty} \frac{\log N(n,n,\rho n,\delta n,\tau n, \tau n)}{n} = -2\log x^*(\delta) - \rho \log y^*(\delta) - \delta \log z^*(\delta) - 2\tau \log w^*(\delta) \,.  
\end{equation}
The total ball size is $|T(\Delta_{n,\rho n}(\tau n), \delta n)| = \sum_{s=0}^{\delta n} N(n,n,\rho n, s,\tau n, \tau n)$. Hence, we have

\begin{proposition}
For fixed $\rho\ge0$, $0 \le \tau \le \rho$, and $0\le\delta\le2$,
\begin{align*}
	\widetilde{T}(\Delta_\rho(\tau),\delta) 
	& = \max_{0\le \delta_1 \le \delta} -2\log x^*(\delta_1) - \rho \log y^*(\delta_1) - \delta_1 \log z^*(\delta_1) - 2\tau \log w^*(\delta_1) \\
	& = 
		-2\log x^*(\delta) - \rho \log y^*(\delta) - \delta \log z^*(\delta) -2\tau \log w^*(\delta) \,.
\end{align*}
\end{proposition}

Note that setting $w^* = 1$, or equivalently $\tau = \tau_{\rm initial} = \frac{\rho^2}{1+\rho}$, to evaluate the expression $2\bCap(\Delta_\rho(\tau)) - \widetilde{T}(\Delta_\rho(\tau),\delta)$, we obtain a similar expression as for the GV bound. But to obtain a better lower bound than the GV bound, we need to evaluate $R_\textsc{mr}(\Delta_\rho,\delta) =  \max_{\tau\in I} \big[2\bCap(\Delta_\rho(\tau))-\widetilde{T}(\Delta_\rho(\tau),\delta)\big]$. For this purpose, we will be considering $\delta$ as a parametric function in terms of $\lambda_1$ and we iterate over $\lambda_1$ to evaluate the GV-MR bound in the following proposition. 

\begin{theorem}
\label{prop:gvmr-L1}
	For fixed $\rho\ge0$, $0 \le \lambda_1^* \le 1$, we have $\alpha(\Delta_{\rho},\delta(\lambda_1^*))\ge R_{\textsc{mr}}(\Delta_{\rho},\delta(\lambda_1^*))$, where
	$$R_{\textsc{mr}}(\Delta_{\rho},\delta(\lambda_1^*)) \triangleq 2\bCap(\Delta_{\rho}(\tau_{\rm opt})) - \widetilde{T}(\Delta_{\rho}(\tau_{\rm opt}),\delta(\lambda_1^*))$$ and
	\begin{align*}
	\tau_{\rm opt} &= \frac{\rho^2}{1+\rho-\lambda_1^*}, \\
	\delta(\lambda_1^*) &= \frac{2\lambda_1^*\rho}{1-(\lambda_1^*)^2+\lambda_1^*\rho}.
	\end{align*}
\end{theorem} 

\begin{proof}
Appendix~\ref{appen_solve_L1_MR}.   
\end{proof}

\subsubsection{Numerical Results}

In this subsection, we plot the GV and GV-MR lower bounds from Theorems~\ref{prop:gv-L1} and \ref{prop:gvmr-L1}.
For comparison purposes, we first derive a sphere-packing bound in Proposition~\ref{prop:sp-L1}.


\vspace{1mm}

\noindent{\em Sphere-Packing Bound}.
For any $\bu\in \Delta_{n,r}$, consider the set $\bu + \Delta_{b,r} = \{\bu+\bw : \bw \in \Delta_{b,r}\} \subseteq \Delta_{n+b,r}$.
It is easy to see that, if $C\subseteq\Delta_{n,r}$ is a code of minimum distance $d$, the sets $\bu + \Delta_{b,r}$ and $\bv + \Delta_{b,r}$, with $b=\lfloor\frac{d-1}{2}\rfloor$, have to be disjoint for any two distinct codewords $\bu,\bv\in C$.
The cardinality of the largest such code can therefore be upper-bounded by
$$
A(\Delta_{n,r}, d) \le \frac{|\Delta_{n+b,r}|}{|\Delta_{b,r}|} \le  \frac{\binom{n+b+r-1}{r-1}}{\binom{b+r-1}{r-1}} , $$
where $b=\lfloor\frac{d-1}{2}\rfloor$.
The following claim is a direct consequence.

\begin{proposition}
\label{prop:sp-L1}
	For fixed $\rho\ge0$ and $0\le\delta\le2$, we have that $\alpha(\Delta_{\rho},\delta)\le R_{\textsc{sp}}(\Delta_{\rho},\delta)$, where
  $$R_{\textsc{sp}}(\Delta_{\rho},\delta) \triangleq (1+\delta/2+\rho) \HH\Big(\frac{\rho}{1+\delta/2+\rho}\Big) - (\delta/2+\rho) \HH\Big(\frac{\rho}{\delta/2+\rho}\Big) .$$
\end{proposition}

\vspace{1mm}

\noindent{\em Lower Bound from Binary Constant Weight Codes}.
In deriving the GV-MR bound, note that we are optimizing over the parameter $\tau$ where $\tau n$ is the number of zeroes in the constrained codewords. 
In the special case where $\tau = \rho - 1$, we observe that the space $\Delta_{n,r}(p)$ is equivalent to constant weight binary codes of length $r$ and weight $n$.
So, for this instance, the bound corresponds to the GV lower bound for constant weight binary codes of length $\rho n$ and weight $n$ (see for example~\cite{Graham1980}).
For completeness, we derive the bound here.

\begin{proposition} 
	For fixed $\rho \ge 1$ and $\tau = \rho-1$, set $\delta_{\rm{max}}= 2\rho/(\rho-1)$. When $0\le\delta\le\delta_{\rm{max}}$, we have
 $$
 \quad \widetilde{T}(\Delta_\rho(\tau), \delta ) =
  \rho\HH\Big(\frac{1}{\rho}\Big) + \HH\Big(\frac{\delta}{2}\Big) + (\rho-1)\HH\Big(\frac{\delta}{2(\rho-1)}\Big).$$
Otherwise, when  $\delta\ge\delta_{\rm{max}}$, we have $\widetilde{T}(\Delta_\rho(\tau),\delta)=2\rho\HH(\frac{1}{\rho})$.
\end{proposition}

\begin{proof}
In case when $\tau = \rho-1$, we have a sum of $\rho n$ terms equal to $n$, while $(\rho-1) n$ terms are $0$. Since every term is nonnegative, these $n$ nonzero terms must all be equal to $1$. Hence, each vector is equivalent to a binary sequence of length $r = \rho n$ and weight $n$. Hence, the total ball size in this case is easily derived as $|T(\Delta_{n,r}, \floor{\delta n})| = \binom{ r}{n}\sum_{i=0}^{\floor{\delta n/2}} \binom{n}{i}\binom{r - n}{i}$ and hence we have that $\widetilde{T}(\Delta_\rho(\tau), \delta ) = \rho\HH\Big(\frac{1}{\rho}\Big) + \max_{0\le \delta_1 \le \delta/2} \HH(\delta_1) + (\rho-1)\HH(\delta_1/(\rho-1))$. It can be easily verified that $\widetilde{T}(\Delta_\rho(\tau), \delta )$ is maximum when $\delta_1 = {\rm min}(\frac{\rho}{\rho-1},\frac{\delta}{2})$.
\begin{align*}
\widetilde{T}(\Delta_\rho(\tau), \delta ) &= 
\begin{cases}
\rho\HH\Big(\frac{1}{\rho}\Big) + \HH\Big(\frac{\delta}{2}\Big) + (\rho-1)\HH\Big(\frac{\delta}{2(\rho-1)}\Big)\,, & \delta\le \delta_{\rm{max}}\,,\\
2\rho\HH(\frac{1}{\rho})\,, & \delta \ge \delta_{\rm{max}}\,.
\end{cases}
\end{align*} 
\end{proof}


\begin{proposition}\label{prop:LB-CWBC}
    For fixed $\rho \ge 1, \tau = \rho - 1$ and $\delta \ge 0$, we have $R_\textsc{gv}(\Delta_\rho(\tau),\delta) =  2\bCap(\Delta_\rho(\tau))-\widetilde{T}(\Delta_\rho(\tau) ,\delta )$.
\end{proposition}

In Figure~\ref{fig:L1}, all these curves are plotted for the case $\rho = 2$. Observe that both GV and GV-MR bounds improve the GV lower bound corresponding to binary constant weight codes.
Note also that the GV bound is positive only for $\delta < 2(1+\rho)/(2+\rho) = 1.5$ (see Proposition~\ref{cor:rate_L1} and Theorem~\ref{prop:gv-L1}), while the GV-MR bound is positive for all $\delta < 2$ (see Theorem~\ref{prop:gvmr-L1}).

\begin{figure}[t!]
	\begin{center}
		\includegraphics[width=0.7\linewidth]{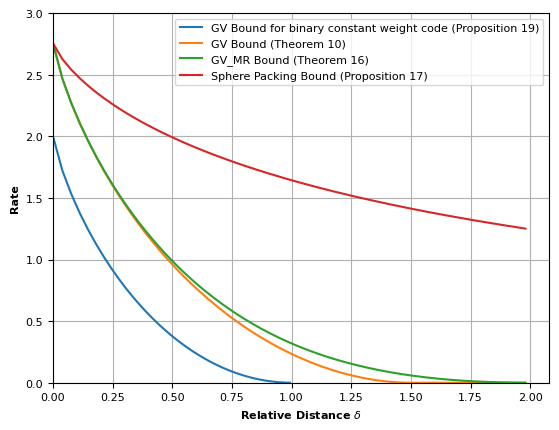}
	\end{center}
	\caption{Bounds on the highest attainable code rate $\alpha(\Delta_{2},\delta)$ for the standard simplex ($\rho=2$).}
	\vspace{-2mm}
	\label{fig:L1}
\end{figure}


\subsection{The Positive Simplex}\label{sec:sticky}

In this subsection, we consider the code space
\begin{equation}
\Delta^{\!+}_{n,r} = \left\{(u_1,u_2, \ldots, u_r)  \in \mathbb{Z}^r: u_i \ge 1, \sum_{i=1}^r u_i = n\right\}
\end{equation}
of cardinality $|\Delta^{\!+}_{n,r}|=\binom{n-1}{r-1}$.
We are again interested in the regime $n\to\infty$, $r=\rho n$, and we denote by $\Delta^{\!+}_{\rho}$ the corresponding family of simplices.
The motivation for studying this space comes from the fact that designing codes correcting a given number of errors in run-preserving channels (i.e., channels that preserve the number of runs in input sequences) such as the sticky-insertion channel is equivalent to designing codes in $\Delta^{\!+}_{n,r}$ under $L_1$ metric.
In this correspondence, the parameter $n$ represents the length of the input sequence, the parameter $r$ the number of runs in that sequence, and $u_i$ the length of the $i$'th run.

\begin{proposition}\label{prop:run-cap}
For fixed $0\le \rho\le 1$, we have
$$
\bCap(\Delta_{\rho}^{\!+}) \triangleq 
\lim_{n\to\infty}\frac{\log\big|\Delta^{\!+}_{n,\lfloor\rho n\rfloor}\big|}{n}=\HH(\rho).
$$
\end{proposition}

One immediately notices that $\Delta^{\!+}_{n,r} = (1, \ldots, 1) + \Delta_{n-r,r}$ and, hence, bounds for codes in $\Delta^{\!+}_{n,r}$ can be directly obtained from the bounds for codes in $\Delta_{n,r}$ given in the previous subsection.
Nonetheless, we state some of the results explicitly so that they can be easily accessed by those interested in coding problems in $\Delta^{\!+}_{n,r}$.
Additionally, we perform optimization of the obtained bounds over the parameter $\rho$, which, in the mentioned application, corresponds to deriving bounds on codes for the sticky insertion channel with no restrictions on the number of runs in input sequences.
Namely, if we consider the code space $\Delta^{\!+}_{n} = \bigcup_{r=0}^n \Delta^{\!+}_{n,r}$, and given that the sticky-insertion channel is run-preserving, implying that an optimal code in $\Delta^{\!+}_{n}$ is a union of optimal codes in $\Delta^{\!+}_{n,r}$ over all $r$, it follows that
\begin{equation}
\alpha(\Delta^{\!+}, \delta) = \max_{0\le\rho\le1}  \alpha(\Delta^{\!+}_\rho, \delta)
\end{equation}
for every $\delta$.

\subsubsection{Evaluation of the GV Bound}
or completeness, we state the following lemma to find the unique solution that corresponds to the generating function of the number of pairs $(\bu,\bv)$ with $L_1$ distance $s$, denoted by $N(n_1,n_2,r,s) = |\{(\bu,\bv) \in \Delta^{\!+}_{n_1,r}\times \Delta^{\!+}_{n_2,r}: D(\bu,\bv) =  s\}|$. The generating function of the quantity $N(n_1,n_2,r,s)$ is given in the following lemma.

\begin{lemma}
	$F(x,y,z,w) = \frac{G(x,y,z,w)}{H(x,y,z,w)}$,
	where $G(x,y,z,w)$ is some multivariate polynomial and 
	{
		\begin{align*}
			H(x,y,z,w) &= (1-x^2)(1-xz) - yx^2(1+xz)\,.  
		\end{align*}
	}
\end{lemma}

As before, we are interested in the pairs residing in the simplex with the same weight $n$, dimension $\rho n$, and distance $\delta n$, thus the quantity of interest is  $N(n,n,\rho n, \delta n)$. By Theorem~\ref{thm:acsv} and Proposition~\ref{prop:barpartial}, with $ r= \rho n, s = \delta n$, we need to solve the following system of equations:
\begin{equation}\label{eq:PositiveSimplex-H}
	H = 0 \text{\quad and\quad}
	\frac{xH_x}{2} = \frac{yH_{y}}{\rho} = \frac{zH_{z}}{\delta}, \;\;\text{where}\;\; H= (1-x^2)(1-xz) - yx^2(1+xz) .
\end{equation}

\begin{lemma}
\label{lem:sol-SI}
For fixed $0 \le \rho \le 1$, the solution of \eqref{eq:PositiveSimplex-H} with respect to $\delta$ is:
		\begin{align*}
			x^*(\delta) &=  \sqrt{1 - \frac{2\rho}{2-\delta}},\\
			z^*(\delta) &= \frac{\sqrt{\rho^2 + \delta^2} - \rho}{x^*\delta},\\
			y^*(\delta) &= 2\frac{\sqrt{\rho^2 + \delta^2}-\delta}{2-\delta - 2\rho}.  
		\end{align*}
\end{lemma}
\begin{proof}
The proof is the same as that of Lemma~\ref{lem:solveH_L1_1}. 
\end{proof}

Applying \eqref{eq:asym_approx}, we get
\begin{equation*}
	\lim_{n \to \infty} \frac{\log N(n,n,\rho n,\delta n)}{n} = -2\log x^*(\delta) - \rho \log y^*(\delta) - \delta \log z^*(\delta)\,.  
\end{equation*}
The total ball size with distance $\delta n$ is $|T(\Delta_{\rho}^{\!+},\delta n)| = \sum_{s=0}^{\delta n}N(n,n,\rho n, s)$, implying that
	\begin{align*}
		\widetilde{T}(\Delta_{\rho}^{\!+},\delta) 
		& = \max_{0\le \delta_1 \le \delta} -2\log x^*(\delta_1) - \rho \log y^*(\delta_1) - \delta_1 \log z^*(\delta_1) \\
		& = \begin{cases}
			-2\log x^*(\delta) - \rho \log y^*(\delta) - \delta \log z^*(\delta)\,, &\delta\le \delta_{\rm{max}} \\
			2\HH(\rho)\,, &\delta \ge \delta_{\rm{max}}
		\end{cases}
	\end{align*}
where $\delta_{\rm{max}}=\frac{2(1-\rho)}{2-\rho}$. 
Consequently, we have the following explicit formula for the asymptotic ball size.
\begin{proposition}
\label{cor:rate_sticky}
	For fixed $0 \le \rho \le 1$, set $\delta_{\rm{max}}=2(1-\rho)/(2-\rho)$.
	When $0\le\delta\le\delta_{\rm{max}}$, we have
	\begin{align*}
		\widetilde{T}(\Delta_{\rho}^{\!+},\delta) =  &-\rho + \delta \log\delta - \rho \log\!\big(\sqrt{\rho^2+\delta^2} - \delta\big) - \delta \log\!\big(\sqrt{\rho^2+\delta^2}-\rho\big) \\
		& +(-1+\rho+\delta/2) \log(2-2\rho-\delta) + (1-\delta/2)\log(2-\delta).
	\end{align*}
Otherwise, when  $\delta\ge\delta_{\rm{max}}$, we have $\widetilde{T}(\Delta_{\rho}^{\!+},\delta)=2\HH(\rho)$.
\end{proposition}

{

\begin{theorem}
\label{prop:gv-si}
For fixed $0 \le \rho \le 1$ and $\delta\ge0$, we have $\alpha(\Delta_{\rho}^{\!+},\delta)\ge R_{\textsc{gv}}(\Delta_{\rho}^{\!+},\delta)$, where
$$R_{\textsc{gv}}(\Delta_{\rho}^{\!+},\delta) \triangleq 2\HH(\rho) - \widetilde{T}(\Delta_{\rho}^{\!+},\delta) .$$
\end{theorem}

We next state the bound obtained by maximizing $R_{\textsc{gv}}(\Delta_{\rho}^{\!+},\delta)$ over $\rho$.
As already mentioned, the resulting function can be directly translated into the GV bound on optimal codes for the sticky insertion channel (having no constraints on the number of runs in input sequences).

\begin{theorem}[GV Bound for Positive Simplex]
\label{prop:gv-si-opt}
For fixed $\delta \ge 0$, we have $\alpha(\Delta^{\!+},\delta)\ge R_{\textsc{gv}}(\Delta^{\!+},\delta)$, where
$$R_{\textsc{gv}}(\Delta^{\!+},\delta) \triangleq 2\HH(\rho) - \widetilde{T}(\Delta^{\!+}_\rho,\delta)$$ and $\rho = \frac{3(2-\delta) - \sqrt{9\delta^2-4\delta+4}}{8} $.
\end{theorem}
\begin{proof}
 Appendix~\ref{appen-optimization}.
\end{proof}

\subsubsection{Evaluation of Marcus and Roth's Improvement of the GV Bound}
Following the approach from \cite{marcus1992improved}, we introduce an additional parameter, which we choose to be the number of vector coordinates with value $1$, then determine the GV bound for this constrained space, and finally optimize the bound over all the values of the new parameter. In particular, we set $\Delta^{\!+}_{n,r}(p) = \{(u_1,u_2,\ldots, u_r) \in \bbZ^{r}: u_i \ge 1, \sum_{i=0}^r u_i=n, |\{i : u_i=1\}|=p\}$. Note that $|\Delta^{\!+}_{n,r}(p)| = \binom{r}{p}\binom{n-r-1}{r-p-1}$. Allowing $p$ to grow linearly with $n$, we denote $p=\tau n$.

\begin{proposition}
\label{prop:run-capMR-si} 
	For fixed $0\le\rho\le1$ and $0 \le \tau \le \rho$, we have $$\bCap(\Delta^{\!+}_\rho(\tau)) \triangleq \lim_{n\to\infty}\frac{\log \big|\Delta^{\!+}_{n,\floor{\rho n}}(\lfloor \tau n \rfloor)\big|}{n}=\rho \HH\Big(\frac{\tau}{\rho}\Big) + (1-\rho)\HH\Big(\frac{\rho-\tau}{1-\rho}\Big).$$
\end{proposition}


To estimate the total ball size, we consider the number of pairs $(\bu,\bv)$ with $L_1$ distance exactly $s$, denoted by $N(n,n,r,s,p,p) = |\{(\bu,\bv) \in \Delta^{\!+}_{n,r}(p)\times \Delta^{\!+}_{n,r}(p): D(\bu,\bv) = s, |\{i : u_i = 1\}| = p, |\{i : v_i = 1\}| = p\}|$. The generating function of the quantity $N(n,n,r,s,p,p)$ is given in the following lemma.

\begin{lemma}\label{lem:general_ge_si}
	$F(x,y,z,w) = \frac{G(x,y,z,w)}{H(x,y,z,w)}$,
	where $G(x,y,z,w)$ is some multivariate polynomial and 
	{
		\begin{align*}
			H(x,y,z,w) &= (1-x^2)(1-xz) - yx^2(w^2(1-xz)(1-x^2) + 2wxz(1-x^2) + x^2(1+xz))\,.  
		\end{align*}
	}
\end{lemma}

As before, the quantity of interest is $N(n,n,\rho n, \delta n, \tau n, \tau n)$. From Theorem~\ref{thm:acsv} and Proposition~\ref{prop:barpartial}, we need to solve the following system of equations.

\begin{equation}\label{eq:SIMR-H}
	H = 0 \text{\quad and\quad}
	\frac{xH_{x}}{2} = \frac{yH_{y}}{\rho} = \frac{zH_{z}}{\delta}  = \frac{wH_{w}}{2\tau}
\end{equation}

Let $\lambda_1 = xz$ and $\lambda_2 = x^2$.

\begin{lemma}\label{lem:solveH_SI}
	For fixed $\rho$ and $\tau$, the solution of the \eqref{eq:SIMR-H} with respect to $\delta$ is:
	\begin{align*}
			x^*(\delta)  &=  \sqrt{\lambda_2^*},\\
			z^*(\delta) &= \frac{\lambda_1^*}{x^*},\\
			w^*(\delta) &= \frac{\lambda_2^*(\delta(1-(\lambda_1^*)^2)-2\lambda_1^*(\rho-\tau))}{\lambda_1^*(1-\lambda_2^{*})(2(\rho - \tau) - \delta(1-\lambda_1^*))}                       , \\
			y^*(\delta) &= \frac{(1-\lambda_1^*)(1-\lambda_2^*)}{\lambda_{2}^{*}((w^*)^2(1-\lambda_1^*)(1-\lambda_2^*) + 2w^*\lambda_1^*(1-\lambda_2^*) + \lambda_2^*(1+\lambda_1^*))}  
		\end{align*}
		where $\lambda_1^*$ is the root of the equation 
		\begin{align*}
	 &(1-\lambda_1)(\delta(1-\lambda_1^2)-2\lambda_1(\rho-\tau))^2(2(1-\rho)-2(\rho-\tau)(1+\lambda_1)-\delta\lambda_1^2) \\& \hspace{2mm}+ \lambda_1^2(1+\lambda_1)(2(\rho-\tau)-\delta(1-\lambda_1))^2(\delta(1-\lambda_1^2)-2\lambda_1(\rho-\tau)-2\tau) = 0.
	 \end{align*}
	 and $\lambda_2^*$ is given by
	 \begin{align*}
	 \lambda_2^*(\delta) = 1 - (1+\lambda_1^*)\frac{2(\rho - \tau) - \delta(1-\lambda_1^*)}{(2(1-\rho)-\delta)}.
		\end{align*}
\end{lemma}

\begin{proof}
 The proof is similar to the proof of Lemma~\ref{lem:solveH_L1} and hence we skip the details. 	
\end{proof}

Applying \eqref{eq:asym_approx}, we have that
\begin{equation}\label{eq:SI-asym_ball}
	\lim_{n \to \infty} \frac{\log N(n,n,\rho n,\delta n,\tau n, \tau n)}{n} = -2\log x^*(\delta) - \rho \log y^*(\delta) - \delta \log z^*(\delta) - 2\tau \log w^*(\delta) \,.  
\end{equation}
The total ball size is $|T(\Delta^{\!+}_\rho(\tau), \delta n)| = \sum_{s=0}^{\delta n} N(n,n,\rho n, s,\tau n, \tau n)$, and hence
\begin{align}\label{eq:Ball-PositiveSimplex}
	\widetilde{T}(\Delta^{\!+}_\rho(\tau),\delta)
	& = \max_{0\le \delta_1 \le \delta} -2\log x^*(\delta_1) - \rho \log y^*(\delta_1) - \delta_1 \log z^*(\delta_1) - 2\tau \log w^*(\delta_1) \notag \\
	& =
		-2\log x^*(\delta) - \rho \log y^*(\delta) - \delta \log z^*(\delta) -2\tau \log w^*(\delta) \,.
\end{align}

Finally, we optimize the GV bound over $\tau$ and parameterize $\delta$ as a function of $\lambda_1^*$ to obtain the GV-MR bound given in the following theorem.

\begin{theorem}
\label{prop:gvmr-SI}
	For fixed $0 \le \lambda_1^* \le 1$, we have $\alpha(\Delta_{\rho}^{\!+},\delta(\lambda_1^*))\ge R_{\textsc{mr}}(\Delta^{\!+}_\rho,\delta(\lambda_1^*))$, where
	$$R_{\textsc{mr}}(\Delta^{\!+}_\rho,\delta(\lambda_1^*)) \triangleq 2\bCap(\Delta^{\!+}_\rho(\tau_{\rm opt})) - \widetilde{T}(\Delta^{\!+}_\rho(\tau_{\rm opt}), \delta(\lambda_1^*))$$ and
	\begin{align*}
		\tau_{\rm opt} &= \frac{\rho^2}{1-\lambda_1^*(1-\rho)}, \\
		\delta(\lambda_1^*) &= \frac{2\lambda_1^*\rho(1-\rho)}{\rho\lambda_1^* + (1-\rho)(1-(\lambda_1^*)^2)} \,.
	\end{align*}
\end{theorem} 
\begin{proof}
Appendix~\ref{appen_solve_sticky}.   
\end{proof}

As with the GV bound, we state below the bound obtained by maximizing $R_{\textsc{mr}}(\Delta^{\!+}_\rho,\delta)$ over $\rho$.

\begin{theorem}[GV Bound for Positive Simplex]
\label{prop:gv-is-mr}
	For fixed $0 \le \lambda_1^* \le 1$, we have $\alpha(\Delta^{\!+},\delta(\lambda_1^*))\ge R_{\textsc{mr}}(\Delta^{\!+},\delta(\lambda_1^*))$, where
	$$R_{\textsc{mr}}(\Delta^{\!+}, \delta(\lambda_1^*)) \triangleq 2\bCap(\Delta^{\!+}_{\rho_{\rm opt}}(\tau_{\rm opt})) - \widetilde{T}(\Delta^{\!+}_{\rho_{\rm opt}}(\tau_{\rm opt}), \delta(\lambda_1^*))$$ and
	\begin{align*}
		\rho_{\rm opt} &= \frac{2\sqrt{1-\lambda_{1}^{*}}}{3\sqrt{1-\lambda_{1}^{*}} + \sqrt{1+3\lambda_{1}^{*}}} ,\\
		\tau_{\rm opt} &= \frac{\rho_{\rm opt}^2}{1-\lambda_1^*(1-\rho_{\rm opt})} ,\\
		\delta(\lambda_1^*) &= \frac{2\lambda_1^*\rho_{\rm opt}(1-\rho_{\rm opt})}{\rho_{\rm opt}\lambda_1^* + (1-\rho_{\rm opt})(1-(\lambda_1^*)^2)} \,.
	\end{align*}
\end{theorem} 
\begin{proof}
Appendix~\ref{appen_solve_sticky}.
\end{proof}

\subsection{The Inverted Simplex}
\label{sec:nabla}

In this subsection, we consider the code space
\begin{equation}
\nabla_{n,r} = \big\{\bu= (u_1,u_2,\ldots, u_r) \in \bbZ^r : 1 \le u_1 < u_2 < \cdots < u_r \le n \big\}
\end{equation}
consisting of vectors whose components are positive, strictly increasing, and not exceeding $n$.
The motivation for studying this space comes from the bit-shift channel, as well as some types of timing channels.
Namely, codes correcting $t$ shifts of $1$'s are equivalently described as codes in $\nabla_{n,r}$ having minimum $L_1$ distance $>\!2t$.
In this correspondence, the parameter $n$ represents the length of the input binary sequence, the parameter $r$ its Hamming weight, and $u_i$ the position of the $i$'th $1$ in that sequence.
As the bit-shift channel does not alter the Hamming weight of the input sequence, one may without loss of generality consider codes for each weight $r$ separately.
After deriving the lower bounds for $\nabla_{n,r}$, one can easily obtain the corresponding bounds for the case with no weight constraints by performing maximization of the bounds over all possible values of $r$.

The set $\nabla_{n,r}$ is also an $r$-dimensional simplex of cardinality $|\nabla_{n,r}| = \binom{n}{r}$. As before, we are interested in the asymptotic regime $n\to\infty$, $r=\rho n$, and we denote by $\nabla_\rho$ the family of simplices satisfying this relation.

\begin{proposition}
\label{prop:run-cap-IS}
	For fixed $0\le\rho\le1$, $$\bCap(\nabla_\rho) \triangleq \lim_{n\to\infty}\frac{\log |\nabla_{n,\floor{\rho n}}|}{n}=\HH(\rho).$$
\end{proposition}

To estimate the total ball size, consider the number of pairs $(\bu,\bv)$ with $L_1$ distance exactly $s$, denoted by $N(n_1,n_2,r,s) = |\{(\bu,\bv) \in \nabla_{n_1,r}\times \nabla_{n_2,r}: D(\bu,\bv) = s\}|$, for which the following recursive relation holds.
 
\begin{lemma}
\label{lem:general_IS}
	\begin{align*}
		N(n_1,n_2,r,s) =\  &\sum_{i \ge 1}N(n_1-i,n_2-i,r-1,s-|n_2-n_1|) \\ 
		&+ \sum_{i \ge 1}\sum_{j \ge 1}N(n_1-i,n_2-i-j,r-1,s-|n_1-n_2 + j|) \\
		&+ \sum_{i \ge 1}\sum_{j \ge 1}N(n_1-i-j,n_2-i,r-1,s-|n_1-n_2 - j|)\,.  
	\end{align*}
\end{lemma}
\begin{proof}
	Let $\bu= (u_1,\ldots, u_r) \in \nabla_{n_1,r}$ and $\bv=(v_1,\ldots,v_r) \in \nabla_{n_2,r}$. We consider truncating the last coordinate $u_r$ and $v_r$. If $u_r=n_1 +1-i,v_r=n_2 +1-i$, we get the first sum where the distance is $s-|u_r-v_r| = s-|n_2-n_1|$. Otherwise, suppose $u_r = n_1 +1-i$ and $v_r = n_2 +1-i-j$, for $i \ge 1$ and $j \ge 1$. Here, the maximum value in any coordinate of the truncated vectors is $n_1-i$ and $n_2 -i-j $, respectively, and their distance is $s-|u_r-v_r|=s - |n_1 - n_2 + j|$. Hence, we get the second term. The last term is obtained similarly when $u_r = n_1-i-j$ and $v_r = n_2-i$.
\end{proof}

Since the total ball size is $|T(\nabla_{n, \rho n}, \delta n)| = \sum_{s=0}^{\delta n} N(n,n,\rho n, s)$, the recursive function of total ball size for $\nabla_{n,r}$ is the same as for $\Delta^{\!+}_{n,r}$.
Therefore, we have the same generating function and the same asymptotic total ball size.
Since the capacity expression for $\nabla_{n,r}$ is also the same as for $\Delta^{\!+}_{n,r}$, the resulting GV and GV-MR bounds are the same as well (Theorems \ref{prop:gv-si-opt} and \ref{prop:gv-is-mr}).

In Figure~\ref{fig:invertedsimplex} we plot the GV and GV-MR lower bounds from Theorems~\ref{prop:gv-si-opt} and~\ref{prop:gv-is-mr}, respectively, and compare them against the lower bound on bit-shift error correcting codes given by Kolesnik and Krachkovsky in~\cite{Kolesnik1994}.

\begin{theorem}[{\cite[Theorem 2]{Kolesnik1994}}]\label{prop:GV-KK}
	For fixed $0\le\delta\le 1/2$, we have that $\alpha(\nabla,\delta) \ge R_{\textsc{kk}}(\nabla,\delta)$, where
	$$R_{\textsc{kk}}(\nabla,\delta) \triangleq 1-\HH(\delta) .$$
\end{theorem} 

\begin{figure}[t!]
	\begin{center}
		\includegraphics[width=0.7\linewidth]{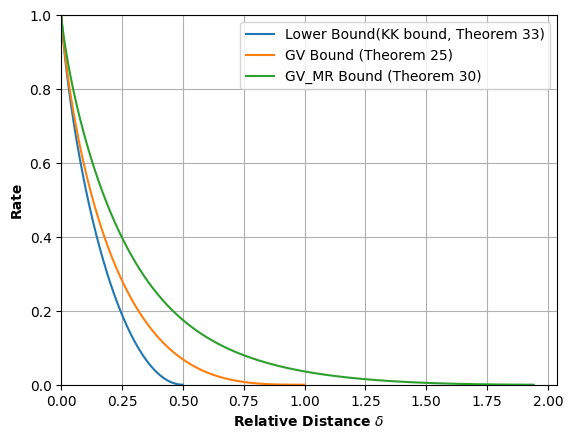}
	\end{center}
	\caption{Lower bounds on the highest attainable code rate for the positive simplex and the inverted simplex.}
	\vspace{-2mm}
	\label{fig:invertedsimplex}
\end{figure}

%% file: hypercube.tex
In this section, we consider the space
\begin{equation}
\bbZ_{q}^{n} = \big\{\bu= (u_1,u_2,\ldots, u_n) \in \bbZ^n: 0 \le u_i \le q-1 \big\} ,
\end{equation}
a discrete hypercube of dimension $n$ and cardinality $|\bbZ_{q}^{n}| = q^n$.
Clearly, $\bCap(\bbZ_{q}) \triangleq \lim_{n\to\infty}\frac{1}{n}\log |\bbZ_{q}^{n}| = \log q$ .

\subsection{Evaluation of the GV Bound}\label{sec:TotalBallSizeHypercube}

To estimate the ball size, we first consider the number of pairs $(\bu,\bv)$ with $L_1$ distance exactly $s$, denoted by $N(n,s) = |\{(\bu,\bv) \in \bbZ_{q}^{n} \times \bbZ_{q}^{n}: D(\bu,\bv) = s\}|$.

\begin{lemma}
\label{lem:general_hypercube}
The following recursion holds:
	\begin{align*}
		N(n,s) &= qN(n-1,s) + 2\sum_{j = 1}^{q-1} (q-j)N(n-1,s-j) \,. 
	\end{align*}
\end{lemma}
\begin{proof}
	Let $\bu= (u_1,\ldots, u_n)$ and $\bv=(v_1,\ldots,v_n) \in N(n,s)$. We consider truncating the last coordinate $u_n$ and $v_n$. If $u_n=v_n=i$ for $i\in \{0,1,\ldots,q-1\}$, we get the first sum where the distance remains the same. Otherwise, $u_n = i$ and $v_n = i+j$ or $u_n = i+j$ and $v_n = i$ for $j \in \{1,2,\ldots,q-1 \}$ and $i \in \{0,1,\ldots,q-j-1\} $. Here, the distance decreases by $|u_n-v_n|=j$. Since there are exactly $(q-j)$ possible pairs with distance $j$, we get the second term.
\end{proof}

The generating function of the bivariate sequence $N(n,s)$, namely $F(x,y)=\sum_{n,s \ge 0}N(n,s)x^{n}y^s$, is given in the following lemma. 
\begin{lemma}\label{lem:general_ge_hypercube}
	$F(x,y) = \frac{1}{H(x,y)}$, where
		\begin{align*}
			H(x,y) &= 1-x \bigg(q+2\sum_{j=1}^{q-1}(q-j)y^j \bigg) \,.  
		\end{align*}
\end{lemma}
\begin{proof}
We have
		\begin{align*}
			F(x,y)
			&= \sum_{n,s \ge 0}N(n,s)x^ny^s \\ 
			&= \sum_{n,s \ge 0}qN(n-1,s)x^ny^s
			   + \sum_{n,s \ge 0}\sum_{j = 1}^{q-1}2(q-j) N(n-1,s-j)x^ny^s \\ 
			&=  1 + xF(x,y) \bigg(q + 2\sum_{j=1}^{q-1}(q-j)y^j \bigg) .
		\end{align*}
\end{proof}

As before, the quantity of interest is $N(n,\delta n)$.
From Theorem~\ref{thm:acsv}, we need to solve the following system of equations:. 
\begin{equation}\label{eq:Hypercube-H}
	H = 0 \text{\quad and\quad}
	xH_{x}  = \frac{yH_{y}}{\delta}. 
\end{equation}

\begin{lemma}\label{lem:solveH_Hypercube_1}
	The solution of the \eqref{eq:Hypercube-H} with respect to $\delta$ is
	{	\begin{align*}
		x^*(\delta) &= \frac{1}{q+2\sum_{j=1}^{q-1}(q-j)(y^*)^j} \,,
		\end{align*}
	and $y^*(\delta)$ is the root of the equation
	\begin{equation}\label{eq:y_hypercube}
		2\sum_{j=1}^{q-1} (q-j)(j-\delta)y^j = q\delta \,.
	\end{equation}
	}
\end{lemma}

\begin{proof}
	From Lemma~\ref{lem:general_ge_hypercube}, we have that $H = 1-x\big(q+2\sum_{j=1}^{q-1}(q-j)y^j\big)$. Thus, differentiating $H$ with respect to $x$ and $y$ and multiplying by $x$ and $y$ respectively, we have that
	\begin{align*}
		xH_x &= -x\bigg(q+2\sum_{j=1}^{q-1}(q-j)y^j\bigg) \\
		yH_y &= -x\bigg(2\sum_{j=1}^{q-1}(q-j)jy^j\bigg) \,.
	\end{align*}
Since $H = 1-x\big(q+2\sum_{j=1}^{q-1}(q-j)y^j\big) = 0$, we have
\begin{align*}
	x(\delta) &= \frac{1}{q+2\sum_{j=1}^{q-1}(q-j)y^j} \,.
\end{align*}
and $xH_x = -1$.
Furthermore, since $\delta xH_x = yH_y$, we have
\begin{align}\label{eq:L1solvehyper}
	\delta = 2x\bigg(\sum_{j=1}^{q-1}(q-j)jy^j\bigg) \,.
\end{align}
By substituting $x$ in \eqref{eq:L1solvehyper}, we obtain the required expression.
\end{proof}

Applying \eqref{eq:asym_approx}, we have that

\begin{equation}\label{eq:asym_ball_hypercube}
	\lim_{n \to \infty} \frac{\log N(n,\delta n)}{n} = -\log x^*(\delta) - \delta \log y^*(\delta)\,.  
\end{equation}
The total ball size is given by $|T(\bbZ^n_q,\delta n)| = \sum_{s=0}^{\delta n} N(n,s)$. Hence, we have that
\begin{align*}
	\widetilde{T}(\bbZ_q, \delta) &= \max_{0\le \delta_1 \le \delta} -\log x^*(\delta_1) - \delta_1 \log y^*(\delta_1) \\
	& = \begin{cases}
		-\log x^*(\delta) - \delta \log y^*(\delta)\,, &\delta\le \delta_{\rm{max}}\,,\\
		2\log q \,, &\delta \ge \delta_{\rm{max}}\,.
	\end{cases}
\end{align*}
Here
\begin{align}
\label{eq:deltamax_hypercube_gv}
\delta_{\rm{max}}=\frac{q^2-1}{3q} ,
\end{align}
which is obtained by setting $y = 1$ in \eqref{eq:y_hypercube}.
We can now state the GV bound.

\begin{theorem}[GV bound for Hypercube]
\label{prop:gv-hypercube}
	For fixed $\delta \ge 0$, we have $\alpha(\bbZ_{q},\delta) \ge R_{\textsc{gv}}(\bbZ_q, \delta)$, where
	$$R_{\textsc{gv}}(\bbZ_q, \delta) \triangleq 2\log q - \widetilde{T}(\bbZ_q, \delta) .$$
\end{theorem}

\subsection{Evaluation of Marcus and Roth's Improvement of the GV Bound}

In this subsection, we introduce an additional parameter $p$ representing the number of components with value $0$. In particular, we set $\bbZ_{q}^{n}(p) = \{(u_1,u_2,\ldots, u_n) \in \bbZ^n: 0 \le u_i \le q-1, |\{i : u_i=0\}|=p\}$.
Allowing $p$ to grow linearly with $n$ we set $p=\tau n$, and we denote by $\bbZ_q(\tau)$ the corresponding family of spaces.
Note that $|\bbZ_{q}^{n}(p)| = \binom{n}{p}(q-1)^{(n-p)}$, and hence the capacity $\bCap(\bbZ_{q}(\tau))$ is given by the following closed formula.

\begin{proposition}\label{prop:run-capMR-hypercube}
	For fixed $ 0 \le \tau \le 1$, we have that $$\bCap(\bbZ_{q}(\tau)) \triangleq \lim_{n\to\infty}\frac{\log |\bbZ_{q}^{n}(\lfloor \tau n \rfloor)|}{n}=(1-\tau)\log(q-1) +  \HH(\tau).$$
\end{proposition}

We perform the same analysis for total ball size to obtain Marcus and Roth's improvement.
Namely, we first consider the number of pairs $(\bu,\bv)$ with $L_1$ distance exactly $s$, denoted by $N(n,s,p_1,p_2) = |\{(\bu,\bv) \in \bbZ_{q}^n(p_1) \times \bbZ_{q}^n(p_2): D(\bu,\bv) = s, |\{i : u_i = 0\}| = p_1, |\{i : v_i = 0\}| = p_2\}|$.

\begin{lemma}
\label{lem:general_hypercube_mr}
The following recursion holds:
	\begin{align*}
		N(n,s,p_1,p_2) =\ &N(n-1,s,p_1-1,p_2-1) + \sum_{j = 1}^{q-1} N(n-1,s-j,p_1-1,p_2) \\
		& + \sum_{j = 1}^{q-1} N(n-1,s-j,p_1,p_2-1) + (q-1)N(n-1,s,p_1,p_2) \\
		& + 2\sum_{j=1}^{q-1}(q-1-j)N(n-1,s-j,p_1,p_2)\,.
	\end{align*}
\end{lemma}

\begin{proof}
	Let $\bu= (u_1,\ldots, u_n)$ and $\bv=(v_1,\ldots,v_n) \in N(n,s)$. We consider truncating the last coordinate $u_n$ and $v_n$. If $u_n = v_n = 0$, we get the first term. If $u_n=v_n=i$ for $i\in \{1,\ldots,q-1\}$, we get the fourth term where the distance remains the same. If $u_n = 0$ and $v_n = j$, or $u_n = j$ and $v_n = 0$, for $j \in \{1,2,\ldots,q-1 \}$, the distance decreases by $|u_n-v_n|=j$ and we get the second and the third term, respectively. Finally, consider $u_n = i$ and $v_n = i+j$, or $u_n = i+j$ and $v_n = i$, for $j \in \{1,2,\ldots,q-1 \}$ and $i \in \{1,2,\ldots,q-j-1\} $. In these cases, the distance decreases by $|u_n-v_n|=j$. Since there are exactly $(q-1-j)$ such pairs with distance $j$, we get the fifth term.
\end{proof}

As before, the quantity of interest is  $N(n, s, p, p)$. Therefore, we define the generating function $F(x,y,w)=\sum_{n,s,p \ge 0}N(n,s,p,p)x^{n}y^sw^{2p}$ with three variables (see Proposition~\ref{prop:barpartial}).

\begin{lemma}\label{lem:general_ge_hypercube_mr}
	$F(x,y,w) = \frac{1}{H(x,y,w)}$,
	where
		\begin{align*}
			H(x,y) &= 1-x\bigg(w^2+(q-1) + 2w\sum_{j=1}^{q-1} y^j +2\sum_{j=1}^{q-1}(q-1-j) y^j\bigg)\,.  
		\end{align*}
\end{lemma}
\begin{proof}
		\begin{align*}
			F(x,y,w)
			&= \sum_{n,s,p \ge 0}N(n,s,p,p)x^ny^sw^p \\ 
			&= \sum_{n,s,p \ge 0}N(n-1,s,p-1,p-1)x^ny^sw^p + \sum_{n,s,p \ge 0}\sum_{j=1}^{q-1}N(n-1,s-j,p-1,p)x^ny^sw^p\\ 
			&\hspace{2mm} + \sum_{n,s,p \ge 0}\sum_{j=1}^{q-1}N(n-1,s-j,p,p-1)x^ny^sw^p + \sum_{n,s,p \ge 0}(q-1)N(n-1,s,p,p)x^ny^sw^p \\ 
			& + \sum_{n,s,p \ge 0} 2 \sum_{j = 1}^{q-1}(q-1-j) N(n-1,s-j,p,p)x^ny^sw^p \\ 
		&=  1 + xF(x,y)\bigg(w^2 + (q-1) + 2w\sum_{j=1}^{q-1}y^j +2\sum_{j=1}^{q-1}(q-1-j)y^j\bigg)\,.
		\end{align*}
\end{proof}

Here, we are interested in $N(n,\floor{\delta n}, \floor{\tau n}, \floor{\tau n})$. From Theorem~\ref{thm:acsv} and Proposition~\ref{prop:barpartial}, we need to solve the following system of equations.
\begin{equation}
\label{eq:HypercubeMR-H}
	H = 0 \text{\quad and\quad}
	xH_{x}  = \frac{yH_{y}}{\delta} = \frac{wH_{w}}{2\tau}.  
\end{equation}


\begin{lemma}\label{lem:solveH_HypercubeMR_1}
The solution of the \eqref{eq:HypercubeMR-H} with respect to $\delta$ is
	\begin{align*}
			x^*(\delta,\tau) &= \frac{1}{w^2 + q-1 + 2w\sum_{j=1}^{q-1}y^j +2\sum_{j=1}^{q-1}(q-1-j)y^j} ,\\ 
			 w^*(\delta,\tau) &= \frac{\delta(q-1) + 2\sum_{j=1}^{q-1}(\delta-j(1-\tau))(q-1-j)y^j}{\sum_{j=1}^{q-1}(2j(1-\tau)-\delta)y^j}\,. 
		\end{align*}
and $y^*(\delta,\tau)$ is the root of the equation
%
\begin{equation}
\label{eq:y_hypercubemr}
	\delta w^2 + 2w\sum_{j=1}^{q-1} (\delta-j)y^j+\delta (q-1) + 2\sum_{j=1}^{q-1}(q-1-j)(\delta-j) y^j  = 0\,.
\end{equation}
\end{lemma}
\begin{proof}
Appendix~\ref{appen_solve_hypercube}.
\end{proof}

Applying \eqref{eq:asym_approx}, we have that
\begin{equation}\label{eq:asym_ball_hypercubeMR}
	\lim_{n \to \infty} \frac{\log N(n,\floor{\delta n}, \floor{\tau n}, \floor{\tau n})}{n} = -\log x^*(\delta,\tau) - \delta \log y^*(\delta,\tau) - 2\tau \log w^*(\delta,\tau)\,.  
\end{equation}
The total ball size is $|T(\bbZ^n_q(\floor{\tau n}),\floor{\delta n})| = \sum_{s=0}^{\floor{\delta n}} N(n,s,\floor{\tau n},\floor{\tau n})$. Hence, we have that
\begin{align*}
	\widetilde{T}(\bbZ_q(\tau), \delta) 
	& = \max_{0\le \delta_1 \le \delta} -\log x^*(\delta_1,\tau) - \delta_1 \log y^*(\delta_1,\tau) - 2\tau \log w^*(\delta_1,\tau) \\
	& = \begin{cases}
		-\log x^*(\delta,\tau) - \delta \log y^*(\delta,\tau) - 2\tau \log w^*(\delta,\tau)\,, &\delta\le \delta_{\rm{max}} \\
		2\bCap(\bbZ_q(\tau)) \,, &\delta \ge \delta_{\rm{max}}
	\end{cases}
\end{align*}
where
\begin{align}
\label{eq:deltamax_hypercube}
\delta_{\rm{max}}=q(1-\tau)\left(1-\frac{(1-\tau)(2q-1)}{3(q-1)}\right)\,.
\end{align}
We have obtained $\delta_{\rm{max}}$ by setting $y = 1$ in \eqref{eq:y_hypercubemr}.
We can now state the GV-MR bound.

\begin{theorem}[GV-MR bound for Hypercube]
\label{prop:gvmr-hypercube}
	For fixed $\delta \ge 0$, we have $\alpha(\bbZ_{q},\delta)\ge R_{\textsc{mr}}(\bbZ_{q}, \delta)$, where
	$$R_{\textsc{mr}}(\bbZ_{q}, \delta) \triangleq 2\bCap(\bbZ_{q}(\tau)) - \widetilde{T}(\bbZ_{q}(\tau), \delta)$$ and $\tau$ is the solution of the equation 
	\begin{equation}\label{eq:tau_hypercubemr}
		\tau (q-1)^2 - (1-\tau)(q-1) - (q-1)\sum_{j=1}^{q-1}(y^*)^j + 2(1-\tau)\sum_{j=1}^{q-1} j(y^*)^j = 0 \, .
	\end{equation}	
\end{theorem}
\begin{proof}
%
Appendix~\ref{appen_solve_hypercube}. 
\end{proof}

The bound obtained in Theorem~\ref{prop:gvmr-hypercube} is positive for $\delta<\delta_{\max}$, where
\begin{equation}
\label{eq:delta_max_mr_hypercube}
\delta_{\max} = \frac{3q(q-1)}{4(2q-1)} .
\end{equation}
This is seen by setting $y = 1$ in \eqref{eq:tau_hypercubemr}, which gives $\tau_{\rm opt} = (q+1)/(4q-2)$, and then substituting $\tau_{\rm opt}$ in \eqref{eq:deltamax_hypercube}.
We note that the standard GV bound (Theorem~\ref{prop:gv-hypercube}) is positive for
$\delta<\frac{q^2-1}{3q}$ (see \eqref{eq:deltamax_hypercube_gv}).
It is easy to show that the expression in \eqref{eq:delta_max_mr_hypercube} is strictly larger than $\frac{q^2-1}{3q}$ for every $q>2$.

\subsection{Numerical Results}

\begin{figure}[t!]
	\begin{center}
		\includegraphics[width=0.7\linewidth]{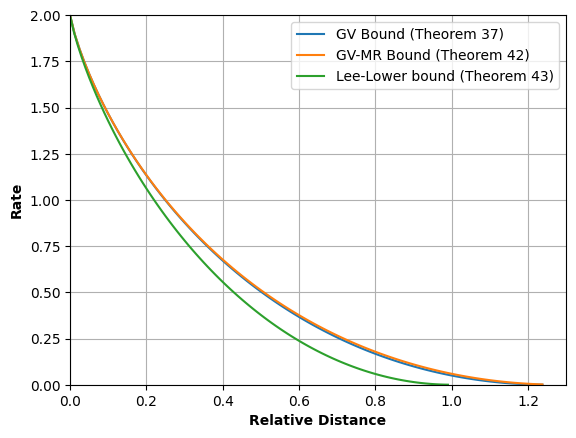}
	\end{center}
	\caption{Lower bounds on the highest attainable code rate $\alpha(\bbZ_{4},\delta)$ for the family of hypercubes $\{\bbZ^n_{4}\}_n$.}
	\vspace{-2mm}
	\label{fig:hypercube}
\end{figure}

In Figure~\ref{fig:hypercube}, we plot the bounds obtained in Theorems~\ref{prop:gv-hypercube} and \ref{prop:gvmr-hypercube} for quaternary alphabet ($q=4$).
We note that, for the GV bound, $\delta_{\max} = (q^2-1)/3q = 5/4$ (see \eqref{eq:deltamax_hypercube_gv}), while for the GV-MR bound, $\delta_{\max} = 3q(q-1)/4(2q-1) = 9/7$ (see \eqref{eq:delta_max_mr_hypercube}).
For comparison purposes, we also plot the GV lower bound for codes in the related {\em Lee distance}.

\subsubsection*{Lower Bound using Lee distance}
A simple upper bound on the total ball size $|T(\bbZ^n_q,d)|$ can be obtained by upper-bounding the volume of an $L_1$-sphere around $\bu$ by the volume of the Lee-sphere around $\bu$.
Namely, the Lee metric is a ``modular version'' of the $L_1$ metric, and a sphere with respect to the former always contains the sphere of the same radius with respect to the latter.
Furthermore, the size of a Lee-sphere is independent of its center, and deriving the corresponding GV bound is straightforward.
Here we give an example for $q=4$ that is used in Figure~\ref{fig:hypercube}.

\begin{theorem}[Lower Bound using Lee-metric]
\label{prop:lb-lee}
	For fixed $0\le\delta\le1$, we have $$\alpha(\bbZ_{4},\delta)\ge (2-\delta) \log (2-\delta) + \delta \log \delta .$$
\end{theorem}
\begin{proof}
The GV bound for Lee-metric codes in $\bbZ^n_q$ is well-known.
Its asymptotic form for $q=4$, namely $(2-\delta) \log (2-\delta) + \delta \log \delta$, was given in \cite[Theorem 6]{Gardy1992}.
By what was said above, the highest attainable rate of $L_1$-metric codes in $\bbZ^n_4$ must also be lower-bounded by this function.
\end{proof}

%% file: GV_copy.bbl

%% file: appendix.tex
\section{Proof of Lemma~\ref{lem:solveH_L1_1}}
\label{appen_solve_L1}

\begin{proof}
We need to find the positive solution of
\begin{equation}
\label{eq:stickyH_L1}
	\begin{split}
		H &= 0 \\
	\frac{xH_x}{2} &= \frac{yH_{y}}{\rho} = \frac{zH_{z}}{\delta}.
	\end{split}
\end{equation}
where $H = (1-x^2)(1-xz) - y(1+xz)$.

Take the partial derivatives as
		\begin{align*}
			\frac{xH_{x}}{2} &= \frac{-\big(2x^2(1-xz)+xz(1-x^2)\big)  - yxz}{2}, \tag{i}\\
			\frac{zH_z}{\delta} &= \frac{-xz(1-x^2) - yxz}{\delta}, \tag{ii}\\ 
			\frac{yH_y}{\rho} &= \frac{-y(1+xz)}{\rho}  \tag{iii}.
		\end{align*}
	Since $H=0$, we substitute $y$ with
	\begin{equation}\label{eq:Ystar_L1}
		y = \frac{(1-x^2)(1-xz)}{(1+xz)}.
	\end{equation}
Therefore, (i)--(iii) simplifies to
		\begin{align*}
			\frac{xH_{x}}{2} &= \frac{-(1+xz)(2x^2+xz-3x^3 z) -(1-x^2)(1-xz)(xz)}{2(1+xz)}\\
			&= -\frac{1}{2(1+xz)}\Big[ (2x^2+xz-3x^3z)(1+xz) - (1-x^2)(x^2z^2)\Big]\\
			&= -\frac{1}{1+xz}\Big[x^2(1-x^2z^2)+xz(1-x^2) \Big], \tag{iv} \\
			\frac{zH_z}{\delta} &= \frac{-xz(1-x^2)}{\delta} -  \frac{xz(1-x^2)(1-xz)}{\delta(1+xz)} \\
			&= -\frac{2xz(1-x^2)}{\delta(1+xz)}, \tag{v} \\
			yH_y &= -\frac{(1-x^2)(1-x^2z^2)}{\rho(1+xz)} \tag{vi}. 
		\end{align*}
	By equating (v) and (vi), we have 
	\begin{align*}
		\frac{2xz(1-x^2)}{\delta} &=  \frac{(1-x^2z^2)(1-x^2)}{\rho} .
	\end{align*}
	Since $x$ is positive, we have
	\begin{equation}\label{eq:Zstar_L1}
		\frac{2xz}{\delta} =  \frac{1-x^2z^2}{\rho} .
	\end{equation}
From \eqref{eq:Zstar_L1}, we have $xz = \frac{-\rho + \sqrt{\rho^2 + \delta^2}}{\delta}$.	
	Lastly, we equate (iv) and (vi) to get 
	\begin{align*}
		x^2(1-x^2z^2)+xz(1-x^2) &= \frac{(1-x^2)(1-x^2z^2)}{\rho}.
	\end{align*}
Recall from \eqref{eq:Zstar_L1} that,
	\begin{align*}
		1-x^2z^2 &= \frac{2\rho}{\delta} xz.
	\end{align*}
Therefore, we have
	\begin{align*}	
		\frac{2\rho}{\delta} (xz)(x^2) + xz(1-x^2) = \frac{2xz(1-x^2)}{\delta} \,.
	\end{align*}	
Solving this equation we get 
\begin{equation}
 \label{eq:Xstar_L1}
		x = \sqrt{1-  \frac{2\rho}{2+2\rho-\delta}}.
\end{equation}
	Therefore, Lemma \ref{lem:solveH_L1_1} results from \eqref{eq:Ystar_L1}, \eqref{eq:Zstar_L1}, and \eqref{eq:Xstar_L1}.
\end{proof}

When $\delta \le \delta_{\max}$, $\widetilde{T}(\Delta_\rho,\delta) = -2\log x - \rho \log y - \delta \log z$, where
\begin{equation*}
	x = \sqrt{1-  \frac{2\rho}{2+2\rho-\delta}}, \quad y = 2\frac{\sqrt{\rho^2+\delta^2}-\delta}{2+2\rho - \delta}, \quad z= \frac{\sqrt{\rho^2+\delta^2}-\rho}{x\delta}.
\end{equation*}
Note that $\delta = \delta_{\rm max}$ corresponds to $z^* = 1$. Therefore, we have that
\begin{align*}
 \delta^2(2-\delta) &= (2+2\rho-\delta)\big(\sqrt{\rho^2+\delta^2}-\rho\big)^2. 
\end{align*}
As $\rho \ge 0$, $0 \le \delta \le 2$, the unique solution is
\begin{equation}
\delta_{\max} = \frac{2(1+\rho)}{2+\rho} .
\end{equation}

\section{Proof of Lemma~\ref{lem:solveH_L1} and Theorem~\ref{prop:gvmr-L1}}
\label{appen_solve_L1_MR} 

\begin{proof}[{Proof of Lemma \ref{lem:solveH_L1}}]
We need to find the positive solution of
\begin{equation}\label{eq:L1MR_H}
	\begin{split}
		H &= 0 \\
		\frac{xH_{x}}{2} = \frac{yH_{y}}{\rho} &= \frac{zH_{z}}{\delta}  = \frac{wH_{w}}{2\tau}
	\end{split}
\end{equation}
where $H = (1-x^2)(1-xz) - y\big(w^2(1-x^2)(1-xz) + 2wxz(1-x^2) + x^2(1+xz)\big)$.

Computing the partial derivatives gives
	{\footnotesize
		\begin{align*}
			H &= (1-x^2)(1-xz) - y\big(w^2(1-x^2)(1-xz) + 2wxz(1-x^2) + x^2(1+xz)\big),  \tag{i} \\
			\frac{xH_{x}}{2} &= \frac{-2x^2(1-xz) - xz(1-x^2) -y\big(w^2((1-x^2)(-xz)-2x^2(1-xz)) +2w(xz(1-x^2)-2x^3z)+2x^2(1+xz)+x^3z \big)}{2}, \tag{ii}\\
			\frac{zH_z}{\delta} &= \frac{-xz(1-x^2) - y\big(-w^2xz(1-x^2)+2wxz(1-x^2)+x^3z \big)}{\delta}, \tag{iii}\\ 
			\frac{yH_y}{\rho} &= \frac{-y\big(w^2(1-x^2)(1-xz) + 2wxz(1-x^2) + x^2(1+xz) \big)}{\rho}  \tag{iv} \\
			\frac{wH_w}{2\tau} &= \frac{-y\big(w^2(1-x^2)(1-xz) + wxz(1-x^2) \big)}{\tau}  \tag{v} \\
		\end{align*}}
Set $\lambda_1 = xz$, $\lambda_2 = x^2$, which transforms the set of equations to
	{\footnotesize
		\begin{align*}
			\frac{xH_{x}}{2} &= \frac{-2\lambda_2(1-\lambda_1) - \lambda_1(1-\lambda_2) -y\big(\lambda_2((1-\lambda_2)(-\lambda_1)-2\lambda_2(1-\lambda_1)) +2w(\lambda_1(1-\lambda_2)-2\lambda_1\lambda_2)+2\lambda_2(1+\lambda_1)+\lambda_1\lambda_2 \big)}{2}, \tag{vi}\\
			\frac{zH_z}{\delta} &= \frac{-\lambda_1(1-\lambda_2) - y\big(-w^2\lambda_1(1-\lambda_2)+2w\lambda_1(1-\lambda_2)+\lambda_1\lambda_2 \big)}{\delta}, \tag{vii}\\ 
			\frac{yH_y}{\rho} &= \frac{-y\big(w^2(1-\lambda_2)(1-\lambda_1) + 2w\lambda_1(1-\lambda_2) + \lambda_2(1+\lambda_1) \big)}{\rho},  \tag{viii} \\
			\frac{wH_w}{2\tau} &= \frac{-y\big(w^2(1-\lambda_2)(1-\lambda_1) + w\lambda_1(1-\lambda_2) \big)}{\tau} \,.  \tag{ix}
		\end{align*}}
Then, since $H=0$, we substitute $y$ with	\begin{equation}\label{eq:Ystar_L1MR}
		y = \frac{(1-\lambda_1)(1-\lambda_2)}{w^2(1-\lambda_1)(1-\lambda_2) + 2w\lambda_1(1-\lambda_2) + \lambda_2(1+\lambda_1)}.
	\end{equation}
By equating (vii), (viii) and (ix), we have
\begin{align}
\label{eq:omega2star_L1MR}
		\lambda_2 &= 1 - (1+\lambda_1)\frac{2(\rho-\tau)-\delta(1-\lambda_1)}{2-\delta}
\end{align}
and
\begin{align}
\label{eq:wstar_L1MR}	
		w &= \frac{\lambda_2(\delta(1-\lambda_1^2)-2\lambda_1(\rho-\tau))}{\lambda_1(1-\lambda_2)(2(\rho-\tau)-\delta(1-\lambda_1))}.
\end{align}
By equating (vi), (vii) and substituting $y,\lambda_2$ and $w$, we have that $\lambda_1^*$ is the root of the equation
{\small
\begin{align}\label{eq:omega1star_L1MR}
	& \delta^2(1+\rho)\lambda_1^5 + (-\delta^2(1+\tau)+4\delta(\rho-\tau)(1+\rho))\lambda_1^4 + (-\delta^2(\rho+2)+2\delta(\rho-\tau)(\rho-\tau-2)+4(\rho-\tau)^2(1+\rho))\lambda_1^3 \notag \\ & + (-\delta^2(\rho-2\tau-2)-2\delta(\rho-\tau)(\rho+\tau+2)-4(\rho-\tau)^2(1-\rho))\lambda_1^2 
	+ (\delta^2 - 4\delta(\rho-\tau)(\rho-\tau-1))\lambda_1 + \delta^2(\rho-\tau-1) = 0 \,.
\end{align}}
Therefore, Lemma \ref{lem:solveH_L1} results from \eqref{eq:Ystar_L1MR}, \eqref{eq:omega2star_L1MR}, \eqref{eq:wstar_L1MR}, and \eqref{eq:omega1star_L1MR}.
\end{proof}

\vspace*{2mm}

\begin{proof}[{Proof of Theorem \ref{prop:gvmr-L1}}]
Note that we already have the generating function $$F(x_{1},y_{1},w_{1}) = \sum_{n \ge 0} |\Delta_{n, r}(p)| x_{1}^{n}y_{1}^{r}w_{1}^{p} = \frac{1-x_{1}}{(1-x_{1})-y_{1}(w_{1}(1-x_{1})+x_{1})}$$ that counts the number of words in $\Delta_{n,r}(p)$. Further, we have the unique solution $w_{1}^* = \frac{\tau(1+\tau-\rho)}{(\rho-\tau)^2}$ from Theorem~\ref{thm:acsv}. Thus, we apply Corollary~\ref{cor:MRsupport} to obtain $\tau_{\rm opt}$, that is $w_1^* = w^*$.

Therefore, we have that
\begin{equation}
\label{eq:condMRL1}
\frac{\rho(1+\tau-\rho)}{(\rho-\tau)^2} = \frac{\lambda_2(\delta(1-\lambda_1^2)-2\lambda_1(\rho-\tau))}{\lambda_1(1-\lambda_2)(2(\rho-\tau)-\delta(1-\lambda_1))}.	
\end{equation}
We substitute $\lambda_2$ from \eqref{eq:omega2star_L1MR} in \eqref{eq:condMRL1} and solve it together with \eqref{eq:omega1star_L1MR} to get, 
\begin{align*}
	\tau_{\rm opt} &= \frac{\rho^2}{1+\rho-\lambda_1^*}, \notag \\ 
	\delta &=  \frac{2\rho\lambda_1^*}{1+\rho\lambda_1^*-(\lambda_1^*)^2}.
\end{align*} 

Note that $\delta = \delta_{\rm max}$ corresponds to $z^* = 1$, which implies that $\lambda_2^* = (\lambda_1^*)^2$. Thus from \eqref{eq:omega2star_L1MR}, we have that
$1 - (1+\lambda_1^*)\frac{2(\rho-\tau)-\delta(1-\lambda_1^*)}{2-\delta} = (\lambda_1^*)^2$, solving which we get
\begin{align*}
 \tau &= \rho + \lambda_1^* - 1 .  
\end{align*}
Equating it to $\tau_{\rm opt}$, we have that $\lambda_1^* = 1$ and therefore $\delta_{\rm max} = 2$.
\end{proof}

\section{Proof of Theorem~\ref{prop:gv-si-opt}}\label{appen-optimization}

It is easy to show that $F(x_{1},y_{1}) = \sum_{n \ge 0} |\Delta^{+}_{n, r}| x_{1}^{n}y_{1}^{r} = \frac{1-x_{1}}{1-x_{1}-x_{1}y_{1}}$. Further, we have the unique solution $y_{1}^* = \frac{\rho}{1-\rho}$ from Theorem~\ref{thm:acsv}. Thus, we apply Corollary~\ref{cor:MRsupport} to obtain $\rho_{\rm opt}$, that is $(y_1^*)^2 = y^*$, where $y^* = 2\frac{\sqrt{\rho^2 + \delta^2}-\delta}{2-\delta - 2\rho} $ from Lemma~\ref{lem:sol-SI}. Therefore, we have that
\begin{align*}
	\frac{\rho^2}{(1-\rho)^2} &= 2\frac{\sqrt{\rho^2 + \delta^2}-\delta}{2-\delta - 2\rho}.
\end{align*}
We solve this equation to get
\begin{align*}
\rho = \frac{3(2-\delta) - \sqrt{9\delta^2-4\delta+4}}{8} .
\end{align*}

{
\section{Proof of Theorem~\ref{prop:gvmr-SI} and Theorem~\ref{prop:gv-is-mr} }\label{appen_solve_sticky}

\begin{proof}[{Proof of Theorem \ref{prop:gvmr-SI}}]
 The following generating function is easy to derive
 \begin{equation}
 \label{eq:genf_deltanrp}
 F(x_{1},y_{1},w_{1}) = \sum_{n \ge 0} |\Delta^{\!+}_{n, r}(p)| x_{1}^{n}y_{1}^{r}w_{1}^{p} = \frac{1-x_{1}}{1-x_{1}-x_{1}y_{1}(w_{1}(1-x_{1})+x_{1})} .
 \end{equation}
 Furthermore, we have the unique solution $w_{1}^* = \frac{\tau(1+\tau-2\rho)}{(\rho-\tau)^2}$ from Theorem~\ref{thm:acsv}. Thus, we apply Corollary~\ref{cor:MRsupport} to obtain $\tau_{\rm opt}$, that is $w_1^* = w^*$.

Note that from Lemma~\ref{lem:solveH_SI}, we have
\begin{align*}
w^*(\delta) &= \frac{\lambda_2^*(\delta(1-(\lambda_1^*)^2)-2\lambda_1^*(\rho-\tau))}{\lambda_1^*(1-\lambda_2^{*})(2(\rho - \tau) - \delta(1-\lambda_1^*))}, \tag{i}\\
y^*(\delta) &= \frac{(1-\lambda_1^*)(1-\lambda_2^*)}{\lambda_2^*((w^*)^2(1-\lambda_1^*)(1-\lambda_2^*) + 2w^*\lambda_1^*(1-\lambda_2^*) + \lambda_2^*(1+\lambda_1^*))},\tag{ii} \\
\lambda_2^*(\delta) &= 1 - (1+\lambda_1^*)\frac{2(\rho - \tau) - \delta(1-\lambda_1^*)}{(2(1-\rho)-\delta)} \tag{iii} 
\end{align*}
and
\begin{align*}
	&(1-\lambda_1)(\delta(1-\lambda_1^2)-2\lambda_1(\rho-\tau))^2(2(1-\rho)-2(\rho-\tau)(1+\lambda_1)-\delta\lambda_1^2) \notag \\& \hspace{2mm}+ \lambda_1^2(1+\lambda_1)(2(\rho-\tau)-\delta(1-\lambda_1))^2(\delta(1-\lambda_1^2)-2\lambda_1(\rho-\tau)-2\tau) = 0. \tag{iv} 
\end{align*}
Therefore, we have
\begin{align}\label{eq:condMRSI-w}
	\frac{\tau(1+\tau-2\rho)}{(\rho-\tau)^2} &= \frac{\lambda_2^*(\delta(1-(\lambda_1^*)^2)-2\lambda_1^*(\rho-\tau))}{\lambda_1^*(1-\lambda_2^{*})(2(\rho - \tau) - \delta(1-\lambda_1^*))}.
\end{align}
We substitute $\lambda_2^*$ from (iii) in \eqref{eq:condMRSI-w} and then solve it together with (iv) to get,
\begin{align*}
\tau_{\rm opt} &= \frac{\rho^2}{1-\lambda_1^*(1-\rho)}, \tag{v}\\ 
\delta(\lambda_1^*) &= \frac{2\lambda_1^*\rho(1-\rho)}{\rho\lambda_1^* + (1-\rho)(1-(\lambda_1^*)^2)}. \tag{vi} 
\end{align*}
Note that we have $\delta = \delta_{\rm max}$ corresponds to $z^* = 1$, which implies that $\lambda_2^* = (\lambda_1^*)^2$. Thus from (iii), we have that
$1 - (1+\lambda_1^*)\frac{2(\rho-\tau)-\delta(1-\lambda_1^*)}{2(1-\rho)-\delta} = (\lambda_1^*)^2$. Solving this, we get
\begin{align*}
	\tau &= \rho - (1-\lambda_1^*)(1-\rho).  
\end{align*}
and thus equating it by $\tau_{\rm opt}$, we have that $\lambda_{1}^{*} = 1$ and therefore $\delta_{\rm max} = 2$.
\end{proof}

\begin{proof}[{Proof of Theorem \ref{prop:gv-is-mr}}]
We need to optimize the result obtained in Theorem~\ref{prop:gvmr-SI} using Theorem~\ref{thm:acsv} and Corollary~\ref{cor:MRsupport}.
Recall the generating function \eqref{eq:genf_deltanrp}.
We have the unique solution $y_{1}^* = \frac{(1-\rho)(\rho-\tau)^2}{\rho(1+\tau-2\rho)^2}$ from Theorem~\ref{thm:acsv}. Thus, we apply Corollary~\ref{cor:MRsupport} to obtain $\rho_{\rm opt}$, that is $y_1^* = y^*$.
We have 
\begin{equation}\label{eq:condMRIS}
	\frac{(1-\rho)(\rho-\tau)^2}{\rho(1+\tau-2\rho)^2} = \frac{(1-\lambda_1^*)(1-\lambda_2^*)}{\lambda_2^*((w^*)^2(1-\lambda_1^*)(1-\lambda_2^*) + 2w^*\lambda_1^*(1-\lambda_2^*) + \lambda_2^*(1+\lambda_1^*))}.
\end{equation}
Now, we first substitute $\lambda_2^*$ from (iii) in \eqref{eq:condMRIS} and then we solve it by substituting $\delta$ and $\tau$ from (v) and (vi) to get,
\begin{align*}
 \rho_{\rm opt} &= \frac{2\sqrt{1-\lambda_{1}^{*}}}{3\sqrt{1-\lambda_{1}^{*}} + \sqrt{1+3\lambda_{1}^{*}}}. \tag{vi}
\end{align*}      
\end{proof}

}

\section{Proof of Lemma~\ref{lem:solveH_HypercubeMR_1} and Theorem~\ref{prop:gvmr-hypercube}}
\label{appen_solve_hypercube}

\begin{proof}[{Proof of Lemma \ref{lem:solveH_HypercubeMR_1}}]
We need to find the positive solution of
\begin{equation}
\label{eq:stickyH_Hypercube}
	\begin{split}
		H &= 0 \\
		xH_x &= \frac{yH_{y}}{\delta} = \frac{wH_{w}}{2\tau}
	\end{split}
\end{equation}
where $H = 1-x\Big(w^2+(q-1) + 2w\sum_{j=1}^{q-1} y^j+2\sum_{j=1}^{q-1}(q-1-j) y^j\Big)$.

Take the partial derivatives as
		\begin{align*}
			xH_{x} &= -x\bigg(w^2+(q-1) + 2w\sum_{j=1}^{q-1} y^j+2\sum_{j=1}^{q-1}(q-1-j) y^j\bigg) = -1, \tag{i}\\
			\frac{yH_y}{\delta} &= \frac{-x}{\delta}\bigg(2w\sum_{j=1}^{q-1} jy^j+2\sum_{j=1}^{q-1}(q-1-j)j y^j\bigg), \tag{ii}\\ 
			\frac{wH_w}{2\tau} &= \frac{-x}{\tau}\bigg(w^2 + w\sum_{j=1}^{q-1} y^j\bigg)  \tag{iii}.
		\end{align*}
Since $H = 0$, we substitute $x$ with	
\begin{equation}
\label{eq:x_hyper}
		x = \frac{1}{w^2+(q-1) + 2w\sum_{j=1}^{q-1} y^j+2\sum_{j=1}^{q-1}(q-1-j) y^j}.
\end{equation}
By equating (i), (ii) and (iii), we have 
	\begin{align*}
		2w\sum_{j=1}^{q-1} jy^j+2\sum_{j=1}^{q-1}(q-1-j)j y^j &= 
		\delta\bigg(w^2+(q-1) + 2w\sum_{j=1}^{q-1} y^j+2\sum_{j=1}^{q-1}(q-1-j) y^j\bigg), \\
		w^2 + w\sum_{j=1}^{q-1} y^j &= \tau\bigg(w^2+(q-1) + 2w\sum_{j=1}^{q-1} y^j+2\sum_{j=1}^{q-1}(q-1-j) y^j\bigg) ,
	\end{align*}
which simplifies to
		\begin{align*} 
		 \delta w^2 + 2w\sum_{j=1}^{q-1} (\delta-j)y^j+\delta (q-1) + 2\sum_{j=1}^{q-1}(q-1-j)(\delta-j) y^j & = 0, \tag{iv}\\
		 (1-\tau) w^2  + (1-2\tau)w\sum_{j=1}^{q-1} y^j -\tau (q-1) -2\tau \sum_{j=1}^{q-1}(q-1-j) y^j & = 0. \tag{v}	
		\end{align*}
Solving (iv) and (v) together, we have
\begin{equation}
\label{eq:wstar_hyper} 
	 	w  = \frac{\delta (q-1) + 2\sum_{j=1}^{q-1}(q-1-j)(\delta - j(1-\tau)) y^j }{\sum_{j=1}^{q-1} (2j(1-\tau)-\delta)y^j}.   
\end{equation}
Therefore, Lemma \ref{lem:solveH_HypercubeMR_1} results from (iv), \eqref{eq:x_hyper} and \eqref{eq:wstar_hyper}.
\end{proof}

\begin{proof}[{Proof of Theorem~\ref{prop:gvmr-hypercube}}]
It is easy to obtain the generating function that counts the number of words in $\bbZ_{q}^{n}(p)$, namely
$$F(x_{1},w_{1}) = \sum_{n \ge 0} |\bbZ_{q}^{n}(p)| x_{1}^{n}w_{1}^{p} = \frac{1}{1-x_{1}(q-1+w_{1})} .$$ Further, we have the unique solution $w_{1}^* = \frac{\tau(q-1)}{1-\tau}$ from Theorem~\ref{thm:acsv}. Thus, we apply Corollary~\ref{cor:MRsupport} to obtain $\tau_{\rm opt}$, that is $w_1^* = w^*$.
We replace $w^*$ by $w_1^*$ in (v) to get, 
\begin{equation}
\label{eq:condMRhyper}
	\tau (q-1)^2 - (1-\tau)(q-1) - (q-1)\sum_{j=1}^{q-1}(y^*)^j + 2(1-\tau)\sum_{j=1}^{q-1} j(y^*)^j = 0 .
\end{equation}
\end{proof}

%% file: GV_ver5.bbl
\begin{thebibliography}{10}
\providecommand{\url}[1]{#1}
\csname url@samestyle\endcsname
\providecommand{\newblock}{\relax}
\providecommand{\bibinfo}[2]{#2}
\providecommand{\BIBentrySTDinterwordspacing}{\spaceskip=0pt\relax}
\providecommand{\BIBentryALTinterwordstretchfactor}{4}
\providecommand{\BIBentryALTinterwordspacing}{\spaceskip=\fontdimen2\font plus
\BIBentryALTinterwordstretchfactor\fontdimen3\font minus
\fontdimen4\font\relax}
\providecommand{\BIBforeignlanguage}[2]{{%
\expandafter\ifx\csname l@#1\endcsname\relax
\typeout{** WARNING: IEEEtran.bst: No hyphenation pattern has been}%
\typeout{** loaded for the language `#1'. Using the pattern for}%
\typeout{** the default language instead.}%
\else
\language=\csname l@#1\endcsname
\fi
#2}}
\providecommand{\BIBdecl}{\relax}
\BIBdecl

\bibitem{bitouze2010}
N. Bitouz\'e, A. Graell i Amat, and E. Rosnes,
``Error correcting coding for a nonsymmetric ternary channel,''
\emph{IEEE Trans. Inf. Theory}, vol.~56, no.~11, pp. 5715--5729, 2010.


\bibitem{dolecek2010repetition}
L.~Dolecek and V.~Anantharam,
``Repetition error correcting sets: Explicit constructions and prefixing methods,''
\emph{SIAM J. Discrete Math.}, vol.~23, no.~4, pp. 2120--2146, 2010.

\bibitem{Gardy1992}
D.~Gardy and P.~Sol\'e,
``Saddle point techniques in asymptotic coding theory,''
in \emph{Lecture Notes in Computer Science, Algebraic Coding}, vol.~573, pp. 75--81, 1992.

\bibitem{Gilbert1952}
E.~N.~Gilbert,
``A comparison of signaling alphabets,''
\emph{Bell Syst. Tech. J.}, vol.~31, no.~3, pp.~504--522, 1952.

\bibitem{keshav2022evaluating}
K.~Goyal and H.~M.~Kiah,
``Evaluating the {G}ilbert--{V}arshamov bound for constrained systems,''
in \emph{Proc. IEEE Int. Symp. Inf. Theory (ISIT)}, pp. 1348--1353, Espoo, Finland, June 2022.

\bibitem{Keshav2023}
K.~Goyal, D.~T.~Dao, H.~M.~Kiah, and M.~Kova{\v{c}}evi{\'c},
``Evaluation of the {G}ilbert--{V}arshamov bound using multivariate analytic combinatorics,''
in \emph{Proc. IEEE Int. Symp. Inf. Theory (ISIT)}, pp. 2458--2463, Taipei, Taiwan, June 2023.

\bibitem{Graham1980}
R.~L.~Graham and N.~J.~A.~Sloane,
``Lower bounds for constant weight codes,''
\emph{IEEE Trans. Inf. Theory}, vol.~26, no.~1, pp. 2658--2668, 1980.

\bibitem{gu1993generalized}
J.~Gu and T.~Fuja,
``A generalized {G}ilbert--{V}arshamov bound derived via analysis of a code-search algorithm,''
\emph{IEEE Trans. Inf. Theory}, vol.~39, no.~3, pp.~1089--1093, 1993.


\bibitem{jain2017}
S. Jain, F. Farnoud, M. Schwartz, and J. Bruck,
``Duplication-correcting codes for data storage in the {DNA} of living organisms,''
\emph{IEEE Trans. Inf. Theory}, vol.~63, no.~8, pp.~4996--5010, 2017.

\bibitem{kolesnik1991generating}
V.~D.~Kolesnik and V.~Y.~Krachkovsky,
``Generating functions and lower bounds on rates for limited error-correcting codes,''
\emph{IEEE Trans. Inf. Theory}, vol.~37, no.~3, pp. 778--788, 1991.

\bibitem{Kolesnik1994}
V.~D.~Kolesnik and V.~Y.~Krachkovsky,
``Lower bounds on achievable rates for limited bitshift correcting codes,''
\emph{IEEE Trans. Inf. Theory}, vol.~40, no.~5, pp. 1443--1458, 1994.

\bibitem{kovavcevic2019RLL}
M.~Kova{\v{c}}evi{\'c},
``Runlength-limited sequences and shift-correcting codes: Asymptotic analysis,''
\emph{IEEE Trans. Inf. Theory}, vol.~65, no.~8, pp.~4804--4814, 2019.

\bibitem{kovavcevic2018multi}
M.~Kova{\v{c}}evi{\'c} and V.~Y.~F.~Tan,
``Codes in the space of multisets---Coding for permutation channels with impairments,''
\emph{IEEE Trans. Inf. Theory}, vol.~64, no.~7, pp. 5156--5169, 2018.

\bibitem{kovavcevic2018}
M.~Kova{\v{c}}evi{\'c} and V.~Y.~F.~Tan,
``Asymptotically optimal codes correcting fixed-length duplication errors in {DNA} storage systems,''
\emph{IEEE Commun. Lett.}, vol.~22, no.~11, pp. 2194--2197, 2018.

\bibitem{kovavcevic2015}
M.~Kova{\v{c}}evi{\'c} and D.~Vukobratovi{\'c},
``Perfect codes in the discrete simplex,''
\emph{Des. Codes Cryptogr.}, vol.~75, no.~1, pp.~81--95, 2015.

\bibitem{kovavcevic2021asymptotic}
M.~Kova{\v{c}}evi{\'c} and D.~Vukobratovi{\'c},
``Asymptotic behavior and typicality properties of runlength-limited sequences,''
\emph{IEEE Trans. Inf. Theory}, vol.~68, no.~3, pp.~1638--1650, 2021.

\bibitem{Kuznetsov1991}
A. V. Kuznetsov and A. J. H. Vinck,
``A coding scheme for single peakshift correction in $(d,k)$-constrained channels,''
\emph{IEEE Trans. Inf. Theory}, vol.~39, no.~4, pp.~1444--1450, 1993.


\bibitem{lenz2018}
A. Lenz, N. J\"{u}nger, and A. Wachter-Zeh,
``Bounds and constructions for multi-symbol duplication error correcting codes,''
preprint \href{https://arxiv.org/abs/1807.02874v1}{arXiv:1807.02874v1}, Jul. 2018.

\bibitem{lenz2021multivariate}
A.~Lenz, S.~Melczer, C.~Rashtchian, and P.~H.~Siegel,
``Multivariate analytic combinatorics for cost constrained channels and subsequence enumeration,''
preprint {\tt arXiv:2111.06105}, 2021.

\bibitem{levenshtein1965}
V. I. Levenshtein,
``Binary codes correcting deletions and insertions of the symbol $ 1 $'' (in Russian),
\emph{Probl. Peredachi Inf.}, vol.~1, no.~1, pp.~12--25, 1965.

\bibitem{levenshtein1966}
V.~I.~Levenshtein,
``Binary codes capable of correcting deletions, insertions, and reversals,''
\emph{Soviet Physics Doklady}, vol.~10, no.~8, pp.~707--710, 1966.

\bibitem{levenshtein1993}
V. I. Levenshtein and A. J. H. Vinck,
``Perfect $(d,k)$-codes capable of correcting single peak-shifts,''
\emph{IEEE Trans. Inf. Theory}, vol.~39, no.~2, pp.~656--662, 1993.

\bibitem{mahdavifar2017}
H.~Mahdavifar and A.~Vardy,
``Asymptotically optimal sticky-insertion-correcting codes with efficient encoding and decoding,''
in \emph{Proc. IEEE Int. Symp. Inf. Theory (ISIT)}, pp.~2683--2687, Aachen, Germany, June 2017.

\bibitem{marcus1992improved}
B.~H.~Marcus and R.~M.~Roth,
``Improved {G}ilbert--{V}arshamov bound for constrained systems,''
\emph{IEEE Trans. Inf. Theory}, vol.~38, no.~4, pp.~1213--1221, 1992.

\bibitem{Melczer2021}
S.~Melczer,
\emph{An Invitation to Analytic Combinatorics: From One to Several Variables},
Texts $\&$ Monographs in Symbolic Computation, Springer International Publishing, 2021.

\bibitem{pemantle2008}
R.~Pemantle and M.~C.~Wilson,
``Twenty combinatorial examples of asymptotics derived from multivariate generating functions,''
\emph{SIAM Review}, vol.~50, no.~2, pp. 199--272, 2008.

\bibitem{shamai1991}
S. Shamai (Shitz) and E. Zehavi,
``Bounds on the capacity of the bit-shift magnetic recording channel,''
\emph{IEEE Trans. Inf. Theory}, vol.~37, no.~3, pp.~863--872, 1991.

\bibitem{Sole2018}
L.~Sok, J.~C.~Belfiore,  P.~Sol\'{e}, and A.~Tchamkerten,
``Lattice codes for deletion and repetition channels,''
\emph{IEEE Trans. Inf. Theory}, vol.~64, no.~3, pp. 1595--1603, 2018.

\bibitem{tallini2010}
L. G. Tallini, N. Elarief, and B. Bose,
``On efficient repetition error correcting codes,''
in \emph{Proc. IEEE Int. Symp. Inf. Theory (ISIT)}, pp.~1012--1016, Austin, TX, USA, June 2010.

\bibitem{Varshamov1957}
R.~R.~Varshamov,
``Estimate of the number of signals in error correcting codes,''
\emph{Doklady Akad. Nauk, SSSR}, vol.~117, pp.~739--741, 1957.

\bibitem{Vasilev1990}
P.~I.~Vasil'ev,
``On $(d,k)$-constrained bitshift correcting block codes,''
in \emph{Proc. 2nd Int. Workshop on Algebraic and Combinatorial Coding Theory}, pp. 215--219, 1990.






\end{thebibliography}
